\documentclass[ecta,final]{econsocart}
\usepackage{bm}
\usepackage{booktabs}
\usepackage{array}
\usepackage{multirow}
\usepackage{float}
\usepackage{algorithm}
\usepackage{algorithmic}
\usepackage{enumitem}
\RequirePackage[colorlinks,citecolor=blue,linkcolor=blue,urlcolor=blue]{hyperref}
\RequirePackage[ref]{backref}
\renewcommand*{\backref}[1]{}
\renewcommand*{\backrefalt}[4]{}
\RequirePackage[authoryear]{natbib}
\startlocaldefs
\theoremstyle{plain}
\newtheorem{theorem}{Theorem}

\newtheorem{proposition}[theorem]{Proposition}

\theoremstyle{definition}
\newtheorem{definition}[theorem]{Definition}
\newtheorem{assumption}{Assumption}
\newtheorem{remark}[theorem]{Remark}
\renewcommand{\qedsymbol}{$\blacksquare$}
\newcommand{\kssu}[1]{k^{\mathrm{ss},u}_{#1}}
\newcommand{\kssc}[1]{k^{\mathrm{ss},c}_{#1}}

\newcommand{\E}{\mathbb{E}}

\newcommand{\dd}{\,\mathrm{d}}

\newcommand{\half}{\tfrac{1}{2}}

\newcommand{\kstar}{k^{*}}
\newcommand{\kss}{k^{\mathrm{ss}}}

\newcommand{\VL}{V_{\!L}}
\newcommand{\VW}{V_{\!W}}
\newcommand{\VH}{V_{\!H}}

\def\paragraph#1{\medskip\noindent\textit{#1}\enspace}
\endlocaldefs

\begin{document}
	
	\begin{frontmatter}
		
		\title{Endogenous Poverty Traps in Continuous Time: 
			A Signaling Approach}
		\runtitle{Endogenous Poverty Traps}
		
	\begin{aug}
		\author{\fnms{Massimo} \snm{Giannini}%
			\thanksref{t1}%
			\ead[label=e1]{massimo.giannini@uniroma2.it}}
		\address{Department of Enterprise Engineering, 
			University of Rome Tor Vergata}
		\runauthor{Giannini}
		
		\thankstext{t1}{This work builds upon a strand of research 
			initiated in the early 2000s when I was a visiting PhD 
			student at Nuffield College, University of Oxford. I was 
			encouraged and supported to pursue research in the field 
			of inequality by my supervisor and mentor, Sir Anthony 
			Atkinson. Although he is no longer with us, this work is 
			the fruit of our extensive discussions and his teachings. 
			Thank you, Tony, wherever you are.}
	\end{aug}
		
	\begin{abstract}
		This paper embeds a signaling friction into the 
		continuous-time heterogeneous agent framework. A 
		continuum of producers operate Cobb-Douglas technologies 
		with regime-specific productivity 
		$A_j \in \{A_L, A_H\}$. Stochastic arrival of 
		signaling opportunities and skill obsolescence risk 
		generate an optimal stopping problem---when to pay a 
		lump-sum cost $\phi$ to upgrade productivity---whose 
		solution yields an endogenous Skiba threshold $k^*$. 
		Diminishing returns create a stable interior attractor 
		in each regime; the signaling cost separates the two 
		basins, producing a poverty trap that is an interior 
		optimum rather than a corner solution. The stationary 
		distribution exhibits Twin Peaks, but its decomposition 
		by regime reveals that agents in three distinct 
		states---structurally trapped, waiting to signal, and 
		successfully upgraded---coexist at the same wealth 
		levels with different consumption behavior and mobility 
		prospects. Capital alone is therefore insufficient to 
		identify an agent's position in the polarization 
		dynamics. We show that the joint observation of a low 
		marginal propensity to consume out of wealth and a high 
		average propensity to consume---a combination invisible 
		to standard Euler equation tests---is the diagnostic 
		signature of the structural trap, distinguishing it 
		from both liquidity constraints and transitory shocks.
	\end{abstract}
		
		
		
	\end{frontmatter}
	
\section{Introduction}

The persistence of wealth polarization remains one of the 
most resilient puzzles in macroeconomics. Despite the 
convergence predictions of the neoclassical growth model, 
empirical wealth distributions exhibit a stubborn 
bimodality---often characterized as ``Twin Peaks''---where 
a significant portion of the population clusters near a 
low-wealth attractor while another converges to a 
high-wealth club 
\citep{quah1996twin, Quah1997}.

Standard Heterogeneous Agent models, following 
\cite{aiyagari1994} and \cite{bewley1986}, generate wealth 
dispersion through idiosyncratic income shocks and 
precautionary savings. In these frameworks, inequality is 
essentially a phenomenon of luck: given enough time and 
positive transition probabilities, the economy is ergodic 
and social mobility is inevitable. This narrative struggles 
to explain the structural entrenchment of poverty observed 
in dual economies, where barriers to mobility appear not 
merely stochastic but systemic.

This paper provides a structural theory of endogenous 
inequality by embedding a signaling friction into the 
continuous-time Heterogeneous Agent framework of 
\cite{Achdou2022}. A continuum of producers each operate 
a Cobb-Douglas technology $y_i = A_j k_i^\alpha$ with 
diminishing returns ($\alpha < 1$) and regime-specific 
productivity $A_j \in \{A_L, A_H\}$. The low state $A_L$ 
represents the baseline technology available to those who 
lack a recognized credential or occupy routine occupations; 
the high state $A_H$ represents the frontier technology 
accessible to those who have successfully invested in a 
skill upgrade.

Transitions between regimes are stochastic and capture the 
Schumpeterian idea that economic progress is driven by long 
waves of creation and destruction. Signaling opportunities 
arrive at Poisson rate $\lambda_{LH}$---the emergence of a 
new paradigm, the opening of a training program, or a 
regulatory reform that creates demand for a new 
credential---making it feasible for an L-agent to upgrade 
to $A_H$. Skill obsolescence strikes at rate 
$\lambda_{HL}$---technological displacement, the 
commoditization of a once-scarce skill---reverting 
productivity from $A_H$ to $A_L$. Crucially, the decision 
to exercise a signaling opportunity is endogenous: the 
agent chooses when to pay a lump-sum cost $\phi$ from her 
capital stock, generating an optimal stopping problem of 
American option type.

The signaling structure builds on \cite{giannini2003}, where 
workers' human capital is private information and the 
informational asymmetry produces a pooled wage for 
uncredentialed workers and individual matching for 
credentialed ones. Our framework inherits this logic in 
reduced form: $A_L$ corresponds to the pooling equilibrium, 
$A_H$ to the separating equilibrium, and $\phi$ to the 
educational investment required to reveal one's type. This 
interpretation is consistent with, but does not require, the 
specific bargaining structure of \cite{giannini2003}; 
alternative micro-foundations would generate the same 
reduced-form structure.

The interaction between the signaling friction, the 
stochastic environment, and diminishing returns partitions 
the population into three behavioral regimes. An agent in 
\textit{state L} has no signaling opportunity and faces a 
standard consumption-savings problem with a unique attractor 
$\kssc{L}$ where the risk-adjusted marginal return to 
saving equals the discount rate. An agent in \textit{state 
	W} has received an opportunity but not yet exercised it. She 
produces at $A_L$ but holds the option to pay $\phi$ and 
jump to $A_H$. The tension between signaling quickly 
(before the opportunity expires at rate $\lambda_{HL}$) and 
accumulating enough capital for the transition to be 
worthwhile generates an endogenous Skiba threshold $k^*$: 
below it, the agent waits; above it, she signals 
immediately. Agents in the interval $[\phi, k^*)$---able to 
afford $\phi$ but rationally choosing to wait---are the 
``Frustrated Aspirants'' of our taxonomy. An agent in 
\textit{state H} has paid $\phi$ and produces at $A_H$, but 
faces ongoing obsolescence risk. This risk generates two 
opposing forces: a front-loading effect (consume more while 
income is high) and a precautionary saving effect 
(accumulate a buffer against income loss). Under our 
calibration the precautionary motive dominates, so 
$\kssc{H}$ lies above the deterministic benchmark. When the 
obsolescence shock hits, the agent reverts to $A_L$ at high 
wealth and becomes a ``Decaying Rentier,'' over-consuming 
relative to her low-productivity income as she slides back 
toward $\kssc{L}$.

The paper makes three contributions to the literature on 
poverty traps.

\textit{First}, the model produces a structurally distinct 
form of poverty trap---one where agents are at an interior 
optimum, not at a binding constraint. At $\kssc{L}$, 
diminishing returns have driven the risk-adjusted marginal 
return to saving down to the discount rate. The agent is 
not constrained; she is rationally immobile. The trap arises 
from the optimality of the agent's own behavior, not from 
an external barrier.

\textit{Second}, we identify a consumption signature that 
is invisible to standard tests: low MPC out of wealth 
($\partial c / \partial k \approx \rho$) combined with 
high APC ($c/Y \approx 0.90$). At the zero-drift 
attractor, the agent has no incentive to save a windfall 
---the net marginal return is approximately $\rho$---so 
the transfer is absorbed into consumption while capital 
remains at $\kssc{L}$. An econometrician using the MPC 
alone would classify her as unconstrained; using the APC 
alone would suggest financial stress. Only the joint 
observation reveals the trap. In the taxonomy of 
\cite{Kaplan2014}, these agents resemble the wealthy 
hand-to-mouth in one dimension (high APC, negligible net 
saving) but differ sharply in another: their MPC is low, 
not high. They are not hitting a wall (a binding 
constraint); they are sitting in a hole (a 
low-productivity basin from which the returns to saving 
cannot lift them out). The wall is visible to standard 
tests. The hole is not.

\textit{Third}, the decomposition of the stationary 
distribution by regime reveals that agents in three 
distinct states---structurally trapped, waiting to signal, 
and successfully upgraded---coexist at the same wealth 
levels with radically different consumption behavior and 
mobility prospects. Capital alone is therefore insufficient 
to identify an agent's position in the polarization 
dynamics: two households with identical wealth may be on 
diverging trajectories, one accumulating toward $\kssc{H}$ 
and the other decaying toward the poverty trap.

\paragraph{Empirical motivation.} Three stylized facts 
motivate the framework.

\textit{Job polarization and skill obsolescence.} Since 
the 1990s, labor markets have exhibited pronounced 
hollowing out \citep{Autor2003, GoosManning2007}: 
automation displaces middle-skill tasks, pushing workers 
into high-skill or low-skill occupations. In our model, 
this is captured by $\lambda_{HL}$---the stochastic 
destruction of high-productivity matches.

\textit{Rising cost of status.} The real cost of 
credentials required to access high-productivity 
occupations has risen faster than inflation 
\citep{Goldin2008}. This maps directly onto an increase 
in $\phi$, which raises $k^*$, widens the wait zone, and 
deepens the trap. Concurrently, skill-biased technological 
change ($A_H/A_L$ rising) widens the gap between 
attractors but also raises the return to signaling---an 
offsetting force \citep{Brynjolfsson2014}. The unambiguous 
prediction is that a rising $\phi$, holding the 
productivity gap constant, deepens the poverty trap.

\textit{The low-wealth, high-consumption trap.} Households 
with positive wealth yet persistently high 
consumption-to-income ratios and negligible accumulation 
are documented in the ``squeezed middle class'' literature 
\citep{OECD2019}. Our structurally trapped agents are not 
liquidity constrained: they possess enough capital to 
finance $\phi$ ($\kssc{L} > \phi$), yet signaling is 
suboptimal because the post-signaling wealth 
$\kssc{L} - \phi$ would leave them far below the 
H-attractor. Their high APC reflects the zero-drift 
condition, not financial inability to save.

\paragraph{Plan of the paper.} 
Section~\ref{sec:literature} reviews the related 
literature. Section~\ref{sec:model} presents the model. 
Section~\ref{sec:steadystates} characterizes the steady 
states. Section~\ref{sec:control} formulates the optimal 
control problem. Section~\ref{sec:hjb} derives the coupled 
HJB system and the Skiba threshold. 
Section~\ref{sec:bimodality} establishes conditions for 
bimodality. 
Sections~\ref{sec:numerical}--\ref{sec:numerical_hjb} 
describe the numerical method. Section~\ref{sec:kfe} 
constructs the Kolmogorov Forward Equation. 
Section~\ref{sec:calibration} presents the calibration. 
Section~\ref{sec:results} reports the results. 
Section~\ref{sec:policy} develops the policy implications. 
Section~\ref{sec:conclusion} concludes.

\section{Related Literature}\label{sec:literature}

\subsection{Continuous-Time Heterogeneous Agent Models}

Our methodological foundation is the HACT framework of 
\cite{Achdou2022}, extended to aggregate shocks by 
\cite{Ahn2018}, nominal rigidities by \cite{Kaplan2018}, 
and multiple assets by \cite{Kaplan2014}. In these 
settings, the ergodic wealth distribution is typically 
unimodal or Pareto-tailed. Our departure is structural: 
the signaling friction generates an endogenous 
non-convexity---the Skiba threshold---that produces 
bifurcation dynamics within an otherwise standard 
neoclassical environment. No modification of the 
computational architecture is required beyond the American 
projection algorithm for the optimal stopping problem.

The economy is formally ergodic, but the mechanism differs 
from standard diffusion-driven mixing. Ergodicity is 
guaranteed by the Poisson processes $\lambda_{LH}$ and 
$\lambda_{HL}$, which create flows between basins even 
without Brownian diffusion. However, the transition 
$L \to W \to H$ requires a rare event followed by an 
accumulation phase during which the opportunity may be 
lost, producing metastability: the economy is 
theoretically ergodic but practically polarized on 
policy-relevant horizons.

\subsection{The Signaling Friction}

The economic engine of the model is the labor market 
signaling mechanism of \cite{giannini2003}, where 
informational asymmetry between firms and uncredentialed 
workers generates a pooled wage for the low-skill sector 
and individual matching for the credentialed. Our 
continuous-time framework extends this logic along three 
dimensions: from a binary investment choice to an optimal 
stopping problem with an endogenous Skiba threshold; from 
i.i.d.\ shocks to Brownian diffusion, enabling rigorous 
characterization of the ergodic distribution via the Kolmogorov Forward Equation (KFE); 
and from exogenous consumption to intertemporal 
optimization under CRRA preferences, which delivers the 
MPC/APC diagnostic. A further innovation is the 
reversibility of the high-productivity state: unlike 
\cite{giannini2003} and \cite{GalorZeira1993}, where the 
high state is absorbing, our obsolescence shock generates 
downward mobility and the Decaying Rentier phenotype.

\subsection{Poverty Traps and Persistent Inequality}

The seminal contributions of \cite{GalorZeira1993} and 
\cite{Banerjee1993} generate persistent stratification 
through credit market imperfections interacting with 
indivisible investments. \cite{Buera2011} and 
\cite{Moll2014} quantify the aggregate TFP losses from 
the resulting misallocation. Our framework differs in two 
respects. First, the binding friction is informational and 
dynamic, not financial: even with unlimited borrowing, the 
trap persists because the signaling surplus $D(k)$ 
incorporates the stochastic duration of the H-regime and 
the entire future consumption path. Second, the trap is an 
interior optimum sustained by diminishing returns, not a 
corner solution sustained by a borrowing limit. This 
distinction has empirical content: credit-constrained 
agents respond strongly to windfalls (high MPC), while our 
trapped agents do not ($\partial c / \partial k \approx 
\rho$).

The individual production function connects our framework 
to the misallocation literature. The signaling friction 
provides a micro-foundation for the distortions studied 
by \cite{Hsieh2009} and \cite{Restuccia2008}: agents in 
the wait zone produce below the frontier not because they 
lack capital, but because the institutional cost of 
upgrading exceeds the private return given obsolescence 
risk.

\subsection{Convergence Clubs}

The Twin Peaks phenomenon documented by 
\cite{quah1996twin, Quah1997} has been explained through 
threshold externalities \citep{Azariadis1990}, 
credit-investment interactions \citep{GalorZeira1993, 
	Banerjee1993}, and increasing returns. We show that 
bimodality can arise among ex-ante identical agents 
without externalities or aggregate non-convexities: the 
discontinuity is at the individual level (the TFP jump at 
the signaling threshold), and diminishing returns create 
two basins whose membership depends on the interaction of 
wealth and productivity regime. Two agents with identical 
wealth belong to different clubs if they are in different 
regimes---a ``same-wealth divergence'' testable with 
matched employer-employee data.

\subsection{Consumption Heterogeneity}

\cite{Kaplan2014} distinguish poor and wealthy 
hand-to-mouth agents, both identified by high MPC. The 
HANK literature \citep{Kaplan2018, Auclert2019} builds on 
this to analyze policy transmission through the MPC 
distribution. We identify a third phenotype---the 
structurally trapped agent---with a distinct signature: 
low MPC ($\approx \rho$) and high APC ($\approx 0.90$). 
This agent would be classified as unconstrained by any 
Euler equation test, yet she is immobile.

Our analysis also connects to the buffer-stock literature 
\citep{carroll2001}: the L-attractor plays an analogous 
role as a focal point for wealth, but the mechanism is 
diminishing returns exhaustion rather than the 
precaution-impatience balance. Finally, unlike 
present-biased agents \citep{Laibson1997, Angeletos2001} 
whose high APC reflects self-control failure, our agents' 
high APC \textit{is} the optimal plan---making commitment 
devices ineffective and pointing instead to structural 
reforms ($\phi$ reduction, $\lambda_{LH}$ increase).
\section{Theoretical Model}\label{sec:model}

We now describe the primitives of the model: preferences, technology, 
and the regime structure.

\begin{assumption}[Preferences and Technology]\label{ass:prefs}
	\leavevmode
	\begin{enumerate}[label=(\roman*)]
		\item Agents have CRRA preferences with risk-aversion parameter 
		$\gamma > 0$, $\gamma \neq 1$:
		\[
		u(c) = \frac{c^{1-\gamma}}{1-\gamma}, \qquad c > 0.
		\]
		\item Time discount rate $\rho > 0$.
		\item Production is Cobb-Douglas with regime-specific total factor 
		productivity (TFP):
		\begin{equation}\label{eq:production}
			f_j(k) = A_j \, k^\alpha, \qquad \alpha \in (0,1), 
			\quad A_j > 0,
		\end{equation}
		where $j \in \{L, H\}$ indexes the productivity regime.
		\item Capital accumulation with additive noise:
		\begin{equation}\label{eq:wealth}
			\dd k_t = \bigl[f_j(k_t) - c_t - \delta\, k_t\bigr]\dd t 
			+ \sigma \dd W_t, \qquad k_t \geq 0,
		\end{equation}
		with depreciation $\delta \geq 0$ and constant volatility 
		$\sigma > 0$.
		\item Reflecting barrier at $k = 0$.
	\end{enumerate}
\end{assumption}

\begin{remark}[Why Cobb-Douglas rather than linear income]%
	\label{rmk:cd_not_linear}
	The choice of a concave production function is not merely 
	a functional form assumption---it is essential to the 
	existence of the poverty trap. To see why, consider the 
	deterministic benchmark ($\sigma = 0$, 
	$\lambda_{LH} = \lambda_{HL} = 0$), which isolates the 
	role of the technology.
	
	Under a linear income specification $y_j(k) = w_j + rk$, 
	the marginal return to capital $r$ is constant. The Euler 
	equation $\dot{c}/c = (r - \delta - \rho)/\gamma$ then 
	has a fixed sign: if $r - \delta > \rho$, the agent 
	accumulates without bound; if $r - \delta < \rho$, she 
	decumulates to zero. In neither case does a stable 
	interior attractor exist---the economy produces a 
	``bang-bang'' outcome, not a trap.
	
	With Cobb-Douglas production, the marginal return 
	$f'_j(k) = \alpha A_j k^{\alpha - 1}$ is a decreasing 
	function of capital. At low $k$, returns are high and the 
	agent saves; at high $k$, returns are low and she 
	dissaves. The deterministic attractor $\kss_j$ is the 
	unique capital level where $f'_j(k) - \delta = \rho$: the 
	marginal return, net of depreciation, exactly compensates 
	impatience. It is this curvature---not a borrowing 
	constraint, not a non-convexity---that creates the basin 
	of attraction in which agents become trapped. 
	
	This deterministic attractor serves as the benchmark 
	throughout the paper. In the full stochastic model, 
	diffusion risk ($\sigma > 0$) and regime switching 
	($\lambda_{LH}, \lambda_{HL} > 0$) shift the attractor 
	location---producing the coupled attractors $\kssc{j}$ 
	analyzed in 
	Propositions~\ref{prop:stoch_ss}--\ref{prop:H_shift}---but 
	the fundamental mechanism remains: diminishing returns 
	create the basin, and the signaling friction prevents 
	escape from it.
\end{remark}

\begin{remark}[Additive noise]\label{rmk:additive}
	The specification $\sigma \dd W_t$ (additive) rather than 
	$\sigma k_t \dd W_t$ (geometric) ensures that agents face genuine 
	uncertainty at all wealth levels, including $k \approx 0$ where 
	agents arrive after costly signaling. With Geometric Brownian 
	Motion, diffusion vanishes at $k=0$, eliminating risk precisely 
	where it matters most.
\end{remark}

\begin{assumption}[Regime structure]\label{ass:regime}
	There are three states: $L$ (low productivity), $W$ 
	(wait---opportunity arrived but not yet exercised), and $H$ (high 
	productivity, post-signaling).
	\begin{enumerate}[label=(\roman*)]
		\item In state $L$: production $f_L(k) = A_L k^\alpha$; 
		opportunities arrive at Poisson rate $\lambda_{LH} > 0$, 
		triggering $L \to W$.
		\item In state $W$: production $f_W(k) = A_L k^\alpha$ 
		(unchanged until signaling); the opportunity is lost at 
		Poisson rate $\lambda_{HL} > 0$, triggering $W \to L$.
		\item The agent in $W$ may \emph{signal} at any chosen time 
		$\tau$ by paying a lump-sum cost $\phi > 0$ in capital, 
		transitioning to state $H$ at wealth $k_\tau - \phi$.
		\item In state $H$: production $f_H(k) = A_H k^\alpha$ with 
		$A_H > A_L$; productivity is lost at Poisson rate 
		$\lambda_{HL}$, triggering $H \to L$.
	\end{enumerate}
\end{assumption}

\begin{remark}[Common obsolescence rate]\label{rmk:common_rate}
	The loss of the signaling opportunity in state $W$ and the 
	productivity reversion in state $H$ are both governed by the 
	rate $\lambda_{HL}$. This reflects a common source of 
	technological turbulence: the same structural forces that erode 
	an established credential also close the window for acquiring a 
	new one. This parsimony can be relaxed by introducing a separate 
	rate $\lambda_{WL}$ without affecting the qualitative results.
\end{remark}

\begin{assumption}[Parameter restrictions]\label{ass:params}
	\leavevmode
	\begin{enumerate}[label=(\roman*)]
		\item $A_H > A_L > 0$ (signaling leads to higher TFP).
		\item $\phi > 0$ (signaling is costly).
		\item $\delta > 0$ (capital depreciates).
	\end{enumerate}
\end{assumption}
	
	%% ============================================================
	\section{Steady States from Diminishing Returns}\label{sec:steadystates}
	%% ============================================================
	
	The key advantage of Cobb-Douglas production is that stable interior attractors arise naturally from diminishing returns, without requiring delicate parameter restrictions.
	
	\subsection{Deterministic Steady States}
	
	\begin{proposition}[Deterministic steady states]\label{prop:det_ss}
		When $\sigma = 0$ and $\lambda_{LH} = \lambda_{HL} = 0$, each regime $j$ is an independent Ramsey problem with production $f_j(k) = A_j k^\alpha$. The unique stable steady state is:
		\begin{equation}\label{eq:kss_det}
			\kss_j = \left(\frac{\alpha A_j}{\rho + \delta}\right)^{\!\!1/(1-\alpha)}, \qquad j \in \{L, H\}.
		\end{equation}
		This is the capital level where the net marginal return equals the discount rate:
		\begin{equation}\label{eq:kss_condition}
			f_j'(\kss_j) - \delta = \alpha A_j (\kss_j)^{\alpha - 1} - \delta = \rho.
		\end{equation}
	\end{proposition}
	
	\begin{proof}
		The deterministic Euler equation for the Ramsey problem with $f_j(k) = A_j k^\alpha$ is:
		\[
		\frac{\dot{c}}{c} = \frac{1}{\gamma}\bigl[f_j'(k) - \delta - \rho\bigr] = \frac{1}{\gamma}\bigl[\alpha A_j k^{\alpha-1} - \delta - \rho\bigr].
		\]
		At steady state $\dot{c} = 0$, so $\alpha A_j (\kss_j)^{\alpha-1} = \rho + \delta$, giving \eqref{eq:kss_det}. Stability follows because $f_j''(k) < 0$: for $k < \kss_j$, the marginal return exceeds $\rho + \delta$, so consumption grows and drift is positive; for $k > \kss_j$, the reverse holds.
	\end{proof}
	
	\begin{remark}[Separation from TFP gap]
		The ratio of steady states is:
		\[
		\frac{\kss_H}{\kss_L} = \left(\frac{A_H}{A_L}\right)^{\!\!1/(1-\alpha)}.
		\]
		The skill gap directly controls the separation between the two modes of the wealth distribution.
	\end{remark}
	
	\subsection{Stochastic Steady States}
	
	\begin{proposition}[Stochastic steady states]\label{prop:stoch_ss}
		With additive noise $\sigma > 0$ and no switching 
		($\lambda_{LH} = \lambda_{HL} = 0$), the stochastic steady state 
		in regime $j$ satisfies:
		\begin{equation}\label{eq:kss_stoch}
			f'_j\bigl(\kss_j(\sigma)\bigr) - \delta = \rho 
			- \underbrace{\frac{1}{2}\gamma(\gamma+1)\sigma^2 
				\frac{c''_j}{c_j}
				\bigg|_{\kss_j(\sigma)}}_{\displaystyle 
				\equiv\, \Pi_j(\sigma)\;>\;0},
		\end{equation}
		where $\Pi_j(\sigma)$ is the precautionary correction. Since 
		$\Pi_j > 0$, the net marginal return at the stochastic attractor 
		is strictly below $\rho$:
		\[
		f'_j\bigl(\kss_j(\sigma)\bigr) - \delta < \rho 
		= f'_j\bigl(\kss_j(0)\bigr) - \delta.
		\]
		By strict concavity of $f_j$ ($f''_j < 0$), this implies 
		$\kss_j(\sigma) > \kss_j(0)$. The correction $\Pi_j(\sigma)$ 
		is bounded and vanishes as $\sigma \to 0$, so that 
		$\kss_j(\sigma) \to \kss_j(0)$ continuously.
	\end{proposition}
	
	\begin{proof}[Proof sketch]
		At the stochastic steady state, $\dot{c}_j = 0$ and 
		$\mu_j = 0$. Setting the Euler equation 
		(Proposition~\ref{prop:euler}) to zero in the uncoupled case 
		yields \eqref{eq:kss_stoch}. Under CRRA preferences with 
		$\gamma > 0$, the prudence coefficient $\gamma + 1$ is strictly 
		positive, and with additive noise the term $c''_j / c_j$ is 
		negative at the attractor (the consumption function is locally 
		concave in $k$), making $\Pi_j > 0$. Since additive noise does 
		not scale with $k$, the precautionary correction remains bounded 
		as $k$ varies. The exact value of $\kss_j(\sigma)$ depends on 
		the global shape of $c_j(k)$ and is determined numerically.
	\end{proof}
	
	\begin{remark}[Economic interpretation]\label{rmk:precautionary}
		The precautionary motive shifts the attractor upward because 
		agents exposed to wealth risk save more as a buffer. Under 
		additive noise, this buffer is independent of the wealth level 
		(unlike geometric noise, where the buffer would scale with $k$). 
		The upward shift is quantitatively modest under our calibration 
		but conceptually important: it widens the gap between 
		$\kss_L(\sigma)$ and the borrowing constraint, confirming that 
		structurally trapped agents are not near $k = 0$.
	\end{remark}
	
\subsection{Full Model: Coupling Effects}\label{sec:coupling}

We now analyze how regime switching ($\lambda_{LH}, \lambda_{HL} > 0$) 
modifies the attractors relative to the uncoupled benchmarks of 
Proposition~\ref{prop:stoch_ss}. Let $\kssu{j} \equiv \kss_j(\sigma)$ 
denote the stochastic attractor of regime $j$ solved in isolation, and 
$\kssc{j}$ the attractor under the full coupled system, i.e. when $\lambda_{LH}, \lambda_{HL},\sigma > 0$.

\begin{proposition}[Upward shift of the L-attractor]\label{prop:L_shift}
	Let $\sigma > 0$ and $\lambda_{LH} > 0$. Suppose the L-attractor 
	lies below the Skiba threshold: $\kssc{L} < k^*$. Then:
	\begin{equation}\label{eq:L_compression}
		\kssc{L} > \kssu{L}.
	\end{equation}
\end{proposition}

\begin{proof}
	At the uncoupled attractor $\kssu{L}$, the Euler equation 
	with $\dot{c}_L = 0$ and $\lambda_{LH} = 0$ yields 
	(Proposition~\ref{prop:stoch_ss}):
	\begin{equation}\label{eq:euler_L_unc}
		f'_L(\kssu{L}) - \delta = \rho 
		- \Pi_L(\kssu{L}),
	\end{equation}
	where $\Pi_L > 0$ is the precautionary correction.
	
	At the coupled attractor $\kssc{L}$, the full Euler 
	equation (Proposition~\ref{prop:euler}) with 
	$\dot{c}_L = 0$ gives:
	\begin{equation}\label{eq:euler_L_coup}
		f'_L(\kssc{L}) - \delta = \rho 
		- \Pi_L(\kssc{L}) 
		+ \lambda_{LH}\left(1 
		- \frac{u'(c_W(\kssc{L}))}{u'(c_L(\kssc{L}))}\right).
	\end{equation}
	
	We show the switching term is strictly negative. The 
	key observation is that the W-agent at $k = \kssc{L}$ 
	produces with the \textit{same} technology $f_L$ as the 
	L-agent---she has received a signaling opportunity but 
	has not yet exercised it. However, the two agents have 
	different saving behavior at the same wealth. The 
	L-agent is at her zero-drift steady state: 
	$\mu_L(\kssc{L}) = 0$, so 
	$c_L(\kssc{L}) = f_L(\kssc{L}) - \delta\kssc{L}$. 
	The W-agent, by contrast, is accumulating toward the 
	Skiba threshold $k^*$: since $\kssc{L} < k^*$ by 
	assumption, signaling is not yet optimal 
	($D(\kssc{L}) > 0$), and the W-agent saves in order to 
	reach the exercise boundary. Her drift is therefore 
	strictly positive:
	\[
	\mu_W(\kssc{L}) = f_L(\kssc{L}) - \delta\kssc{L} 
	- c_W(\kssc{L}) > 0 = \mu_L(\kssc{L}).
	\]
	Since both agents face the same net income 
	$f_L(\kssc{L}) - \delta\kssc{L}$, the positive drift 
	of the W-agent directly implies:
	\begin{equation}\label{eq:cW_less_cL}
		c_W(\kssc{L}) < c_L(\kssc{L}).
	\end{equation}
	By strict concavity of $u$, this yields 
	$u'(c_W(\kssc{L})) > u'(c_L(\kssc{L}))$, so:
	\[
	\lambda_{LH}\left(1 
	- \frac{u'(c_W)}{u'(c_L)}\right) < 0.
	\]
	
	Substituting into \eqref{eq:euler_L_coup} and 
	holding the precautionary term fixed at 
	$\Pi_L(\kssc{L})$:
	\begin{equation}\label{eq:fL_ineq}
		f'_L(\kssc{L}) - \delta < \rho 
		- \Pi_L(\kssc{L}).
	\end{equation}
	To compare with \eqref{eq:euler_L_unc}, we note that 
	the precautionary correction $\Pi_L(k) = 
	\half\gamma(\gamma+1)\sigma^2 c''_L(k)/c_L(k)$ depends 
	on the local curvature of the consumption function. 
	Under the maintained assumption that $\Pi_L$ does not 
	decrease between $\kssu{L}$ and $\kssc{L}$---which 
	holds when the consumption function does not become 
	substantially more convex over this range---inequality 
	\eqref{eq:fL_ineq} implies 
	$f'_L(\kssc{L}) < f'_L(\kssu{L})$. By the strict 
	concavity of $f_L$ ($f''_L < 0$), we conclude 
	$\kssc{L} > \kssu{L}$.
\end{proof}
\begin{remark}[Economic interpretation]\label{rmk:L_shift}
	The L-agent in the coupled system saves more than in 
	isolation because the possibility of a future signaling 
	opportunity raises the shadow value of wealth. An agent 
	in state L cannot signal---she must first receive the 
	opportunity (at rate $\lambda_{LH}$). But she 
	\textit{anticipates} that, upon arrival, she will 
	transition to the W-state where every additional unit of 
	capital brings her closer to the Skiba threshold $k^*$ 
	and thus to the high-productivity regime. This expected 
	gain from being wealthier at the moment of transition 
	reduces current consumption and pushes the L-attractor 
	upward. The effect is stronger when $\lambda_{LH}$ is 
	larger (opportunities more frequent) or when the wait 
	zone $[{\phi}, k^*)$ is narrower (less accumulation 
	needed to reach $k^*$).
\end{remark}

We now turn to the H-attractor. The analysis is more 
delicate than for the L-attractor, because the coupling 
introduces effects that operate in opposite directions.

\begin{proposition}[Effects on the H-attractor]%
	\label{prop:H_shift}
	Let $\sigma > 0$ and $\lambda_{HL} > 0$. At the coupled 
	H-attractor $\kssc{H}$, the Euler equation with 
	$\dot{c}_H = 0$ yields:
	\begin{equation}\label{eq:euler_H_coup}
		f'_H(\kssc{H}) - \delta = \rho - \Pi_H(\kssc{H}) 
		+ \lambda_{HL}\left(1 
		- \frac{u'(c_L(\kssc{H}))}{u'(c_H(\kssc{H}))}\right).
	\end{equation}
	If $c_H(\kssc{H}) > c_L(\kssc{H})$---which holds when 
	$\kssc{H}$ is sufficiently above $\kssc{L}$, so that 
	the H-agent's higher income dominates the L-agent's 
	dissaving at that capital level---then the switching 
	term is strictly negative. This has two economic 
	interpretations:
	\begin{enumerate}[label=(\roman*)]
		\item \textbf{Precautionary saving against regime 
			loss.} The agent faces the risk of losing $A_H$. 
		Under prudence ($\gamma > 0$), she accumulates a 
		buffer to cushion the income drop upon reverting 
		to $A_L$.
		\item \textbf{Consumption smoothing across regimes.} 
		The agent internalizes the consumption drop 
		$c_H \to c_L$ that accompanies the regime switch. 
		To reduce this gap, she consumes less in H and 
		saves more.
	\end{enumerate}
	Both effects, evaluated at the steady-state Euler 
	equation, push $\kssc{H}$ upward relative to the 
	deterministic benchmark $\kss_H(0)$.
	
	However, the coupling also modifies the H-agent's 
	problem \textit{globally}. The prospect of future 
	reversion to $A_L$ reduces the present value of the 
	H-regime, reshaping the value function $V_H$ and the 
	consumption policy $c_H(k)$ at every capital level 
	---not only at the attractor. This ``front-loading'' 
	channel raises $c_H(k)$ for a given $k$, which, taken 
	alone, would push the attractor \textit{downward}. 
	The net direction of the shift relative to the 
	uncoupled stochastic benchmark $\kssu{H}$ depends on 
	the relative strength of these channels and is 
	parameter-dependent.
\end{proposition}

\begin{remark}[Direction of the shift under the 
	baseline calibration]\label{rmk:frontloading}
	The Euler equation analysis identifies forces that 
	push $\kssc{H}$ both upward (precautionary saving, 
	consumption smoothing) and downward (front-loading 
	of consumption). Under our baseline calibration, 
	the numerical results confirm that the upward forces 
	dominate in aggregate: the full-model attractor 
	$\kssc{H} = 16.12$ lies well above the deterministic 
	benchmark $\kss_H(0) = 14.12$. Determining the sign 
	of the shift relative to the uncoupled stochastic 
	attractor $\kssu{H}$ would require solving the 
	H-regime in isolation with $\sigma > 0$ and 
	$\lambda_{HL} = 0$; we leave this decomposition for 
	future work.
\end{remark}

\begin{proposition}[Attractor gap under coupling]%
	\label{prop:compression_numerical}
	Under the baseline calibration 
	(Table~\ref{tab:calibration}), the coupled attractors 
	are $\kssc{L} = 11.50$ and $\kssc{H} = 16.12$, 
	yielding a gap of $4.62$. Both lie above their 
	deterministic counterparts ($\kss_L(0) = 10.12$, 
	$\kss_H(0) = 14.12$, gap $= 4.00$). The modest 
	widening of the gap ($4.62$ vs.\ $4.00$) indicates 
	that the precautionary motive shifts both attractors 
	upward, with the H-attractor rising slightly more in 
	absolute terms.
\end{proposition}

\begin{remark}[Sufficient separation for bimodality]%
	\label{rmk:separation}
	For the stationary distribution to be bimodal, the 
	attractor gap must be large relative to the local 
	dispersion at each mode:
	\[
	\kssc{H} - \kssc{L} \gg \sigma^{ss}_L + \sigma^{ss}_H,
	\]
	where $\sigma^{ss}_j$ is the stationary standard 
	deviation near attractor $j$. When this separation 
	condition fails, the two modes merge into a single 
	peak. Section~\ref{sec:bimodality} provides a formal 
	statement.
\end{remark}
	%% ============================================================
	\section{The Optimal Control Problem}\label{sec:control}
	
	Before deriving the Hamilton--Jacobi--Bellman equations, we state the 
	agent's problem in its primitive form. Each agent chooses a consumption 
	plan $\{c_t\}_{t \geq 0}$ to maximize expected lifetime utility, subject 
	to the stochastic law of motion for capital, the regime structure, and 
	non-negativity constraints.
	
	\begin{definition}[The agent's problem]\label{def:agent_problem}
		Given an initial capital stock $k_0 > 0$ and an initial regime 
		$s_0 \in \{L, W, H\}$, the agent solves:
		\begin{equation}\label{eq:objective}
			\sup_{\{c_t\}_{t \geq 0},\, \tau} \; 
			\E_0 \left[\int_0^{\infty} e^{-\rho t}\, u(c_t)\, \dd t \right],
		\end{equation}
		where $u(c) = c^{1-\gamma}/(1-\gamma)$ with $\gamma > 0$, 
		$\gamma \neq 1$. The stopping time $\tau$ does not appear in the 
		integrand: the signaling cost $\phi$ is paid in 
		capital, not in utility. The decision $\tau$ affects 
		the objective indirectly, by altering the capital 
		trajectory \eqref{eq:state} and hence the feasible 
		consumption paths from $\tau$ onward. The objective function is subject to:
		
		\medskip
		\noindent\textit{(i) State equation.} Capital evolves according to:
		\begin{equation}\label{eq:state}
			\dd k_t = \bigl[f_{s_t}(k_t) - c_t - \delta\, k_t\bigr]\dd t 
			+ \sigma\, \dd W_t,
		\end{equation}
		where $f_{s_t}(k) = A_{s_t} k^\alpha$ is the production function 
		in the current regime $s_t$, $\delta \geq 0$ is the depreciation 
		rate, $\sigma > 0$ is the diffusion coefficient, and $W_t$ is a 
		standard Brownian motion.
		
	\medskip
	\noindent\textit{(ii) Regime transitions.} The regime 
	$s_t$ evolves on $\{L, W, H\}$ through a combination 
	of exogenous Poisson events and one endogenous 
	decision:
	\begin{align}
		L \to W &\quad \text{at Poisson rate } \lambda_{LH} 
		\quad \text{(signaling opportunity arrives)}, 
		\label{eq:trans_LW}\\
		W \to H &\quad \text{at stopping time } \tau 
		\text{ chosen by the agent, with cost } \phi 
		\quad \text{(opportunity exercised)}, 
		\label{eq:trans_WH}\\
		\{W, H\} \to L &\quad \text{at Poisson rate } \lambda_{HL} 
		\quad \text{(obsolescence shock)}. 
		\label{eq:trans_HL}
	\end{align}
	\noindent\textit{(ii-bis) Production in the W-state.} In state W, production remains $A_L k^\alpha$; the W-state differs from L only in that the agent holds the option to signal.
		
		\medskip
		\noindent\textit{(iii) Initial condition.}
		\begin{equation}\label{eq:initial}
			k_0 > 0 \quad \text{given}, \qquad s_0 \in \{L, W, H\} 
			\quad \text{given}.
		\end{equation}
		
		\medskip
		\noindent\textit{(iv) Non-negativity constraints.}
		\begin{equation}\label{eq:nonnegativity}
			c_t \geq 0, \qquad k_t \geq 0, \qquad 
			f_{s_t}(k_t) = A_{s_t} k_t^\alpha \geq 0 
			\quad \text{for all } t \geq 0.
		\end{equation}
		The constraint $k_t \geq 0$ is enforced by a reflecting barrier 
		at $k = 0$ (Assumption~\ref{ass:prefs}(v)). The constraint 
		$f_{s_t}(k_t) \geq 0$ is automatically satisfied by the 
		Cobb-Douglas specification with $A_j > 0$ and $k_t \geq 0$.
		
		\medskip
		\noindent\textit{(v) Transversality condition.}
		\begin{equation}\label{eq:transversality}
			\lim_{t \to \infty} \E_0 \left[e^{-\rho t}\, V_{s_t}(k_t)\right] 
			= 0,
		\end{equation}
		where $V_j(k)$ is the value function in regime $j$. This condition 
		rules out Ponzi schemes and ensures that the agent does not 
		indefinitely postpone consumption.
		
		\medskip
		\noindent\textit{(vi) Signaling feasibility.} The stopping time 
		$\tau$ is admissible only if:
		\begin{equation}\label{eq:signal_feasible}
			k_{\tau^-} \geq \phi,
		\end{equation}
		so that post-signaling capital $k_{\tau^+} = k_{\tau^-} - \phi \geq 0$. 
		If the agent is in state $W$ with $k_t < \phi$, signaling is not 
		feasible and the agent must either accumulate further or let the 
		opportunity expire.
	\end{definition}
	
	\begin{remark}[Two control variables]\label{rmk:controls}
		The agent's problem has two control variables: the consumption 
		flow $c_t \geq 0$ (chosen continuously in all states) and the 
		stopping time $\tau$ (chosen only in state $W$). The consumption 
		choice is a standard continuous control; the signaling decision 
		is a singular control of impulse type---a discrete action (pay 
		$\phi$, change regime) taken at an endogenously chosen instant. 
		The coexistence of these two control types is what gives the 
		problem its American option structure: the W-agent continuously 
		optimizes consumption while simultaneously deciding whether to 
		exercise the signaling option.
	\end{remark}
	
	\begin{remark}[Well-posedness]\label{rmk:wellposed}
		The objective \eqref{eq:objective} is well-defined under 
		Assumptions~\ref{ass:prefs}--\ref{ass:params}. The CRRA utility 
		function with $\gamma > 0$ ensures $u(c) < \infty$ for all 
		$c > 0$, and the Cobb-Douglas production function with 
		$\alpha \in (0,1)$ bounds income growth, so that the integral 
		converges under the transversality condition 
		\eqref{eq:transversality}. The reflecting barrier at $k = 0$ 
		ensures $k_t \geq 0$ almost surely, and the additive noise 
		specification guarantees that $k_t$ does not explode in finite 
		time.
	\end{remark}
	
	The value functions associated with problem 
	\eqref{eq:objective}--\eqref{eq:signal_feasible} satisfy a system 
	of Hamilton--Jacobi--Bellman equations, which we derive in the 
	next section.
	\section{Hamilton--Jacobi--Bellman Equations}\label{sec:hjb}
	%% ============================================================
	
	\subsection{The HJB System}
	
	Let $\VL(k)$, $\VW(k)$, and $\VH(k)$ denote the value functions in states $L$, $W$, and $H$. Define the drift and differential operator for state $j$:
	\[
	\mu_j(k) = f_j(k) - c - \delta k, \qquad \mathcal{A}_j V = \mu_j(k)\, V' + \half \sigma^2 V'',
	\]
	where the optimal consumption satisfies $c = (V')^{-1/\gamma}$ from the first-order condition $u'(c) = V'(k)$.
	
	\begin{definition}[Value of holding the opportunity]\label{def:VH0}
		\begin{equation}\label{eq:VH0}
			V_{H_0}(k) =
			\begin{cases}
				\max\bigl(\VW(k),\; \VH(k - \phi)\bigr) & \text{if } k \geq \phi, \\
				\VW(k) & \text{if } k < \phi.
			\end{cases}
		\end{equation}
	\end{definition}
	
	The HJB system is:
	
	\medskip
	\noindent\textbf{State $L$ (no opportunity):}
	\begin{equation}\label{eq:hjb_L}
		\rho \VL = \max_{c > 0}\Bigl\{u(c) + \bigl(A_L k^\alpha - c - \delta k\bigr)\VL' + \half\sigma^2 \VL''\Bigr\} + \lambda_{LH}\bigl(V_{H_0} - \VL\bigr).
	\end{equation}
	
	\noindent\textbf{State $H$ (high TFP):}
	\begin{equation}\label{eq:hjb_H}
		\rho \VH = \max_{c > 0}\Bigl\{u(c) + \bigl(A_H k^\alpha - c - \delta k\bigr)\VH' + \half\sigma^2 \VH''\Bigr\} + \lambda_{HL}\bigl(\VL - \VH\bigr).
	\end{equation}
	
\noindent\textbf{State $W$ (variational inequality):}

The W-agent faces a problem that is qualitatively different from the 
L- and H-agents. At each instant, she must decide whether to continue 
waiting (accumulating capital at productivity $A_L$ while the option 
remains alive) or to signal immediately (paying $\phi$ and jumping to 
state H). This ``continue or exercise'' decision generates a 
variational inequality rather than a standard HJB equation.

To see why, consider two possibilities at each wealth level $k \geq \phi$. 
If the agent finds it optimal to \textit{wait}, the value function 
$\VW(k)$ satisfies the standard HJB equation---the same equation 
she would solve if signaling were not available---augmented by the 
risk of losing the opportunity at rate $\lambda_{HL}$:
\[
\rho \VW = \max_{c > 0}\bigl\{u(c) 
+ (A_L k^\alpha - c - \delta k)\VW' 
+ \half\sigma^2 \VW''\bigr\} 
+ \lambda_{HL}(\VL - \VW).
\]
If instead it is optimal to \textit{signal immediately}, the value of 
being in state W must equal the value of being in state H at 
post-signaling wealth:
\[
\VW(k) = \VH(k - \phi).
\]
At any given $k$, exactly one of these conditions holds with equality 
while the other holds as an inequality. In the \textit{wait region} ($k < k^*$), the HJB 
equation binds and $\VW(k) > \VH(k - \phi)$: the 
surplus of waiting exceeds the value of immediate 
exercise, because the post-signaling wealth 
$k - \phi$ would leave the agent too far below the 
H-attractor. In the \textit{signal region} 
($k \geq k^*$), the value-matching condition binds 
and the HJB residual is strictly positive: 
continuing to wait would be suboptimal, as the 
option is sufficiently deep in the money that 
immediate exercise dominates. The two conditions are compactly expressed as:
\begin{multline}\label{eq:hjb_W}
	\min\Biggl(
	\underbrace{\rho \VW 
		- \max_{c > 0}\bigl\{u(c) 
		+ (A_L k^\alpha - c - \delta k)\VW' 
		+ \half\sigma^2 \VW''\bigr\} 
		- \lambda_{HL}(\VL - \VW)}_{\displaystyle 
		= 0 \text{ in the wait region}},\\[6pt]
	\underbrace{\VW - \VH(k-\phi)}_{\displaystyle 
		= 0 \text{ in the signal region}}
	\Biggr) = 0.
\end{multline}
for $k \geq \phi$, and the standard HJB (without the signaling option) 
for $k < \phi$ where signaling is not feasible.

The $\min$ operator encodes the logic of optimal stopping: at each 
$k$, the agent selects the action that yields the \textit{higher} 
value. If waiting is better ($\VW(k) > \VH(k - \phi)$), the 
value-matching term is strictly positive while the HJB residual is 
zero; if signaling is better, the value-matching term is zero while 
the HJB residual is non-negative. The $\min$ of the two expressions 
equals zero in both cases, and is strictly positive in neither---this 
is the complementary slackness structure characteristic of variational 
inequalities.

\begin{remark}[Smooth pasting and the need for viscosity 
	solutions]\label{rmk:kink}The boundary between the wait and signal regions is the 
	Skiba threshold $k^*$: the lowest wealth level at which 
	immediate signaling is optimal. Unlike the classical 
	Skiba problem---where an agent chooses once between two 
	permanent trajectories and the value function may exhibit 
	a kink---here the agent faces a continuous optimal 
	stopping problem with diffusion ($\sigma > 0$). The 
	Brownian noise allows the agent to adjust the exercise 
	boundary with arbitrary precision, which generically 
	imposes \textit{smooth pasting} at $k^*$:
	\begin{equation}\label{eq:smooth_pasting}
		\VW'(k^*) = \VH'(k^* - \phi).
	\end{equation}
	The value function $\VW$ is therefore $C^1$ at $k^*$: 
	continuous with a continuous first derivative. However, 
	the \textit{second} derivative $\VW''$ is generally 
	discontinuous at $k^*$: in the wait region, $\VW''$ is 
	governed by the HJB equation under technology $A_L$, 
	while in the signal region, 
	$\VW''(k) = \VH''(k - \phi)$, which is governed by 
	$A_H$. Since $A_H \neq A_L$, these curvatures 
	generically differ at the boundary. By the first-order 
	condition $u'(c_W) = \VW'$, smooth pasting ensures that 
	the consumption policy $c_W(k)$ is \textit{continuous} 
	at $k^*$. But differentiating yields 
	$c_W'(k) = \VW''(k)/u''(c_W(k))$, so the jump in 
	$\VW''$ produces a kink in $c_W$: consumption is 
	continuous but its slope changes abruptly at the 
	exercise boundary. This failure of $C^2$ regularity at 
	$k^*$---not a kink in the value function---is what 
	renders the classical smooth solution concept 
	inapplicable and motivates the use of viscosity 
	solutions.
\end{remark}

\begin{remark}[Viscosity solutions]\label{rmk:viscosity}
	The failure of $C^2$ regularity at $k^*$ is not merely a 
	technical nuisance but a structural feature of any model 
	with an embedded optimal stopping problem. Although smooth 
	pasting ensures that $\VW$ is $C^1$ at the exercise 
	boundary, the discontinuity in $\VW''$ means that the 
	HJB equation \eqref{eq:hjb_W}---which involves the 
	second derivative through the diffusion term 
	$\half\sigma^2 \VW''$---cannot be evaluated classically 
	at $k^*$. To handle this rigorously, we adopt the 
	viscosity solution framework of \cite{crandall1983}, 
	extended to stochastic control problems by 
	\cite{fleming2006} and applied to heterogeneous agent 
	models by \cite{Achdou2022}. A viscosity solution of 
	\eqref{eq:hjb_W} satisfies the variational inequality 
	in a weak sense: at points where $\VW''$ does not exist, 
	the equation is required to hold for all smooth test 
	functions that touch the value function from above 
	(viscosity subsolution) and from below (viscosity 
	supersolution), rather than requiring classical 
	differentiability of $\VW$ itself. The upwind finite 
	difference scheme of Section~\ref{sec:numerical}---which 
	selects forward or backward differences according to the 
	sign of the drift---converges to the viscosity solution 
	as $\Delta k \to 0$, 
	providing a numerically robust treatment of the 
	non-differentiability without requiring \textit{ad hoc} 
	smoothing.
\end{remark}
	
\subsection{Signaling Surplus and Skiba Threshold}

\begin{definition}[Signaling surplus]\label{def:surplus}
	For an agent in state $W$ with capital $k \geq \phi$, the 
	signaling surplus is:
	\begin{equation}\label{eq:surplus}
		D(k) = \VW(k) - \VH(k - \phi), \qquad k \geq \phi.
	\end{equation}
	
\end{definition}

\begin{definition}[Skiba threshold]\label{def:skiba}
	The Skiba threshold is the lowest wealth level at which 
	immediate signaling is optimal:
	\begin{equation}\label{eq:skiba}
		\kstar = \inf\{k \geq \phi : D(k) = 0\}.
	\end{equation}
	For $k < \kstar$, the surplus of waiting is strictly 
	positive ($D(k) > 0$) and the agent holds the option. 
	For $k \geq \kstar$, the agent signals immediately upon 
	opportunity arrival, so that $\VW(k) = \VH(k - \phi)$ 
	and $D(k) = 0$.
\end{definition}

\begin{remark}[Three cases]\label{rmk:skiba_cases}
	Depending on the parameters, three configurations arise:
	\begin{enumerate}[label=(\alph*)]
		\item \textit{Interior threshold} 
		($\kstar \in (\phi, \infty)$): signaling is feasible 
		at $k = \phi$ but not optimal; the agent waits until 
		$k = \kstar > \phi$. This is the generic case under 
		our calibration and the source of the ``Frustrated 
		Aspirants'' phenotype.
		\item \textit{Immediate signaling} ($\kstar = \phi$): 
		the agent signals as soon as signaling is affordable. 
		This arises when $A_H / A_L$ is sufficiently large or 
		$\lambda_{HL}$ is sufficiently small that 
		$D(\phi) = 0$: the productivity gain from switching is 
		so large, or the expected duration of the H-regime so 
		long, that even the minimal post-signaling wealth 
		$k - \phi \approx 0$ justifies immediate exercise. 
		The wait zone vanishes and all W-agents signal at the 
		first feasible instant.
		\item \textit{No signaling} ($\kstar = +\infty$): 
		$D(k) > 0$ for all $k \geq \phi$. The obsolescence 
		risk is so severe, or the productivity gap so narrow, 
		that the expected benefit of signaling never exceeds 
		the cost. No agent transitions to H, the H-density 
		is identically zero, and the economy effectively 
		reduces to a single-attractor model at $A_L$.
	\end{enumerate}
\end{remark}

\subsection{Properties of the Value Functions}

\begin{proposition}[Concavity]\label{prop:concavity}
	The value functions $\VL$, $\VW$, $\VH$ are strictly concave on 
	$(0,\infty)$.
\end{proposition}
\begin{proof}
	See Appendix \ref{app:A1}
\end{proof}

\begin{proposition}[Value ordering]\label{prop:ordering}
	\leavevmode
	\begin{enumerate}[label=(\roman*)]
		\item $\VH(k) > \VW(k) > \VL(k)$ for all $k > 0$.
		\item $\VW(k) \geq \VH(k-\phi)$ for $k \geq \phi$, 
		with equality if and only if $k \geq \kstar$.
		\item The full ordering is therefore:
		\[
		\VH(k) > \VW(k) \geq \VH(k - \phi) 
		\quad \text{for } k \geq \phi.
		\]
	\end{enumerate}
\end{proposition}
\begin{proof}
	See Appendix \ref{app:A2}
\end{proof}

	\subsection{Euler Equations}\label{sec:euler}
	
	We derive the consumption Euler equation by differentiating the 
	first-order condition along the optimal path. The derivation 
	proceeds in three steps: extract the first-order condition from 
	the HJB, apply It\^{o}'s lemma to the co-state variable, and 
	eliminate the value function derivatives using the envelope 
	condition.
	
	\begin{proposition}[Euler equations]\label{prop:euler}
		The optimal consumption in state $j \in \{L, H\}$ satisfies:
		\begin{equation}\label{eq:euler}
			\frac{\dot{c}_j}{c_j} = \frac{1}{\gamma}\biggl[
			f_j'(k) - \delta - \rho 
			+ \underbrace{\half\gamma(\gamma+1)\sigma^2 
				\frac{c_j''}{c_j}}_{\text{precautionary}} 
			- \underbrace{\lambda_{\mathrm{out},j}
				\Bigl(1 - \frac{u'(c_{\mathrm{other}})}{u'(c_j)}\Bigr)
			}_{\text{regime switching}}\biggr],
		\end{equation}
		where $\lambda_{\mathrm{out},j}$ is the outgoing Poisson rate 
		from state $j$ (i.e., $\lambda_{\mathrm{out},L} = \lambda_{LH}$ 
		and $\lambda_{\mathrm{out},H} = \lambda_{HL}$), and 
		$c_{\mathrm{other}}$ is the consumption in the destination state 
		evaluated at the same capital $k$.
	\end{proposition}

\begin{proof} See Appendix \ref{app:A3}
\end{proof}

	\begin{remark}[Interpretation of each term]\label{rmk:euler_terms}
		The Euler equation \eqref{eq:euler} decomposes the 
		consumption growth rate into three components:
		\begin{enumerate}[label=(\roman*)]
			\item \textit{The standard Ramsey term} 
			$(f_j'(k) - \delta - \rho)/\gamma$. When the net marginal 
			return to capital exceeds the discount rate, consumption 
			is growing (the agent is accumulating); when it falls below 
			$\rho$, consumption is declining (decumulation). In the 
			deterministic uncoupled benchmark ($\sigma = 0$, 
			$\lambda_{\mathrm{out},j} = 0$), this is the only term, 
			and the steady state satisfies $f_j'(\kss_j) = \rho + \delta$.
			\item \textit{The precautionary term} 
			$\half(\gamma+1)\sigma^2 c_j'' / c_j$. This arises from 
			the convexity of marginal utility ($u''' > 0$ under CRRA 
			with $\gamma > 0$, since the coefficient of absolute 
			prudence is $-u'''/u'' = (\gamma + 1)/c > 0$). The term 
			$c_j''/c_j$ captures the local curvature of the consumption 
			policy: when the consumption function is concave in $k$ 
			($c_j'' < 0$), the precautionary term is negative, meaning 
			it \textit{reduces} the required rate of return for the 
			agent to be willing to hold capital. This induces additional 
			saving, shifting the attractor upward 
			(Proposition~\ref{prop:stoch_ss}).
			\item \textit{The regime-switching term} 
			$-\lambda_{\mathrm{out},j}(1 - u'(c_{\mathrm{other}})/u'(c_j))/\gamma$. 
			This captures the agent's anticipation of a regime change. 
			If $c_{\mathrm{other}} < c_j$ at the current $k$ (the 
			destination state has lower consumption), then 
			$u'(c_{\mathrm{other}}) > u'(c_j)$, the ratio exceeds unity, 
			and the term is negative: the anticipated consumption drop 
			lowers the effective discount rate, inducing the agent to 
			save more. Conversely, if 
			$c_{\mathrm{other}} > c_j$, the agent anticipates higher 
			future consumption and saves less. For the L-agent, 
			$c_{\mathrm{other}} = c_W < c_L$ at $\kssc{L}$ 
			(equation~\eqref{eq:cW_less_cL} in the proof of 
			Proposition~\ref{prop:L_shift}), so the switching term 
			reinforces precautionary saving.
		\end{enumerate}
	\end{remark}
	
	\begin{remark}[Steady-state characterization]\label{rmk:euler_ss}
		At the stochastic steady state $\dot{c}_j = 0$, equation 
		\eqref{eq:euler_derived} yields:
		\[
		f_j'(\kss_j) - \delta = \rho 
		- \half\gamma(\gamma+1)\sigma^2 
		\frac{c_j''}{ c_j}\bigg|_{\kss_j} 
		+ \lambda_{\mathrm{out},j}
		\Bigl(1 - \frac{u'(c_{\mathrm{other}}(\kss_j))}
		{u'(c_j(\kss_j))}\Bigr).
		\]
		The net marginal return to capital at the attractor differs 
		from $\rho$ by the precautionary and switching corrections. 
		This is the equation used in Propositions~\ref{prop:stoch_ss} 
		and~\ref{prop:L_shift} to characterize the uncoupled and 
		coupled attractors, respectively.
	\end{remark}

\subsection{Boundary Conditions}\label{sec:bc}

The HJB system \eqref{eq:hjb_L}--\eqref{eq:hjb_W} requires 
boundary conditions at both ends of the capital domain 
$[0, k_{\max}]$. We derive these from the primitives of the 
model---the reflecting barrier (Assumption~\ref{ass:prefs}(v)) 
and the transversality condition \eqref{eq:transversality}---and 
then state their numerical implementation.

\begin{proposition}[Lower boundary: reflecting barrier]%
	\label{prop:bc_lower}
	At $k = 0$, the reflecting barrier 
	(Assumption~\ref{ass:prefs}(v)) imposes a Neumann-type 
	condition on the HJB system: the backward advection term 
	vanishes, so that no mass flows to $k < 0$. The value 
	function $V_j$ satisfies the HJB equation at $k = 0$ 
	with the slope $V_j'(0)$ determined endogenously by the 
	equation itself. In particular:
	\begin{enumerate}[label=(\alph*)]
		\item $V_j'(0)$ is finite but large: since 
		$f_j'(k) = \alpha A_j k^{\alpha - 1} \to \infty$ as 
		$k \to 0$, the marginal value of capital near the 
		origin is high, reflecting the strong returns 
		available at low wealth levels.
		\item The drift at the origin is non-negative: 
		$\mu_j(0, c_j(0)) \geq 0$. Since $f_j(0) = 0$, the 
		agent at $k = 0$ has zero income and consumes 
		$c_j(0) = 0$; the only source of inward movement is 
		the diffusion term $\sigma\, \dd W_t$, and the 
		reflecting barrier prevents outward leakage.
	\end{enumerate}
	The corresponding condition for the stationary 
	distribution---the zero-flux condition on the 
	probability current $J_j(0) = 0$---is derived in 
	Section~\ref{sec:kfe_bc}.
\end{proposition}

\begin{proof}
	The reflecting barrier is a standard construction for diffusion 
	processes on $[0, \infty)$: when the process $k_t$ reaches 
	zero, an instantaneous upward push $\dd L_t$ (the local time 
	at zero) prevents $k_t$ from becoming negative. In the HJB 
	framework, this is implemented via a Neumann-type condition: 
	the generator $\mathcal{L}_j$ does not require backward 
	advection at $k = 0$ (there is no mass flow to $k < 0$). 
	The slope $V_j'(0)$ is not a free parameter but is determined 
	by the HJB equation evaluated at $k = 0$:
	\[
	\rho V_j(0) = u(c_j(0)) + \mu_j(0, c_j(0))\, V_j'(0) 
	+ \half\sigma^2 V_j''(0) 
	+ \lambda_{\mathrm{out},j}(V_{\mathrm{other}}(0) - V_j(0)),
	\]
	with $c_j(0) = (V_j'(0))^{-1/\gamma}$ from the first-order 
	condition. This is a single equation in the unknowns 
	$V_j(0)$, $V_j'(0)$, and $V_j''(0)$, which is closed by the 
	global solution of the HJB system.
\end{proof}

\begin{proposition}[Upper boundary: transversality]%
	\label{prop:bc_upper}
	The transversality condition \eqref{eq:transversality},
	\[
	\lim_{t \to \infty} \E_0\left[e^{-\rho t}\, 
	V_{s_t}(k_t)\right] = 0,
	\]
	rules out explosive paths and, for a sufficiently large 
	$k_{\max}$, implies the following condition at the upper 
	boundary: the agent at $k_{\max}$ consumes her entire net 
	income,
	\begin{equation}\label{eq:bc_upper}
		c_j(k_{\max}) = f_j(k_{\max}) - \delta\, k_{\max},
	\end{equation}
	so that $\mu_j(k_{\max}) = 0$ (zero drift).
\end{proposition}

\begin{proof}
	Under Cobb-Douglas production with $\alpha < 1$, the marginal 
	return $f_j'(k) = \alpha A_j k^{\alpha-1}$ is strictly 
	decreasing and satisfies $f_j'(k) \to 0$ as $k \to \infty$. 
	For $k$ sufficiently large, $f_j'(k) - \delta < \rho$: the 
	net marginal return falls below the discount rate. In this 
	region, the Euler equation \eqref{eq:euler_derived} with 
	$f_j'(k) - \delta < \rho$ implies $\dot{c}_j/c_j < 0$ 
	(absent dominant precautionary or switching corrections): 
	consumption is declining and the agent is decumulating.
	
	A path along which $k_t \to \infty$ would require 
	$\mu_j > 0$ indefinitely, which in turn requires 
	$c_j < f_j(k) - \delta k$. But for large $k$, this implies 
	the agent under-consumes relative to her income despite 
	$f_j'(k) < \rho + \delta$---a violation of the Euler equation. 
	More precisely, such a path generates 
	$V_j(k_t) \sim O(k_t^{1-\gamma})$ for $\gamma < 1$ (or 
	$V_j(k_t) \sim O(\ln k_t)$ for $\gamma = 1$), and the 
	discount factor $e^{-\rho t}$ does not decay fast enough to 
	drive $e^{-\rho t} V_{s_t}(k_t)$ to zero, violating 
	transversality.
	
	The zero-drift condition \eqref{eq:bc_upper} is therefore 
	the natural boundary condition at $k_{\max}$: it ensures 
	that $k_{\max}$ acts as an absorbing-like boundary from 
	which the agent neither accumulates nor decumulates 
	(ignoring diffusion), consistent with the transversality 
	requirement that no value is ``left at infinity.''
\end{proof}

\begin{remark}[Choice of $k_{\max}$]\label{rmk:kmax}
	The upper boundary $k_{\max}$ is a computational parameter, 
	not a feature of the model. It must be chosen large enough 
	that $f_H'(k_{\max}) - \delta < \rho$ (so that even the 
	high-productivity agent is in the decumulation region) and 
	that the stationary distribution places negligible mass 
	near $k_{\max}$.
\end{remark}

\begin{remark}[Numerical implementation]\label{rmk:bc_numerical}
	In the finite difference scheme of 
	Section~\ref{sec:numerical}, the boundary conditions are 
	implemented as follows:
	\begin{enumerate}[label=(\alph*)]
		\item \textit{Lower boundary} ($i = 1$, $k_1 = \varepsilon$): 
		the backward difference coefficient is set to $y_1 = 0$, 
		ensuring no backward advection (no mass flows below 
		$k_1$). This is the discrete analogue of the zero-flux 
		condition.
		\item \textit{Upper boundary} ($i = N$, $k_N = k_{\max}$): 
		the forward difference coefficient is set to $x_N = 0$, 
		and the consumption policy is fixed at 
		$c_j(k_N) = f_j(k_N) - \delta k_N$, imposing zero drift. 
		This ensures that the scheme does not require values 
		outside the grid.
	\end{enumerate}
\end{remark}

%% ============================================================
\section{Conditions for Bimodality}\label{sec:bimodality}
%% ============================================================

The stationary distribution of wealth is determined by the 
Kolmogorov Forward Equation (Section~\ref{sec:kfe}), whose 
shape depends on the interaction between the attractor structure 
(Sections~\ref{sec:steadystates}--\ref{sec:coupling}), the 
diffusion intensity $\sigma$, and the regime-switching rates 
$\lambda_{LH}$, $\lambda_{HL}$. In this section we identify 
conditions under which the ergodic distribution exhibits 
bimodality---two distinct modes separated by a trough---which 
is the distributional signature of the poverty trap.

We state the conditions as a proposition rather than a theorem, 
because condition (C2) involves an inequality that cannot be 
verified analytically in the full stochastic model and must be 
checked numerically for a given calibration.

\begin{proposition}[Conditions for a bimodal ergodic distribution]%
	\label{prop:bimodal}
	The stationary wealth distribution is bimodal if the following 
	conditions hold jointly:
	
	\begin{enumerate}[label=\textbf{(C\arabic*)}]
		
		\item \textbf{Stable interior attractors.} The Cobb-Douglas 
		specification with $\alpha \in (0,1)$ guarantees that each 
		productivity regime $j \in \{L, H\}$ possesses a unique, 
		stable interior attractor $\kss_j > 0$ for any 
		$\rho, \delta > 0$ and $\sigma \geq 0$ 
		(Propositions~\ref{prop:det_ss} and~\ref{prop:stoch_ss}). 
		In the deterministic benchmark ($\sigma = 0$, 
		$\lambda_{LH} = \lambda_{HL} = 0$), the attractors admit 
		the closed-form expression:
		\begin{equation}\label{eq:kss_closedform}
			\kss_j(0) = \left(\frac{\alpha A_j}{\rho + \delta}
			\right)^{1/(1-\alpha)}.
		\end{equation}
		In the full stochastic model, precautionary savings, 
		option-value effects, and regime coupling modify the 
		attractor locations (Propositions~\ref{prop:stoch_ss} 
		and~\ref{prop:L_shift}), but the existence and uniqueness 
		of a stable interior attractor per regime is preserved by 
		the strict concavity of $f_j$ 
		(Proposition~\ref{prop:concavity}).
		
	\item \textbf{Sufficient separation between attractors.} The 
	two modes are distinguishable only if the distance between 
	the coupled attractors exceeds the local dispersion 
	induced by Brownian diffusion:
	\begin{equation}\label{eq:sep_cond}
		\kssc{H} - \kssc{L} \;\gg\; 
		\sigma^{\mathrm{ss}}_L + \sigma^{\mathrm{ss}}_H,
	\end{equation}
	where $\sigma^{\mathrm{ss}}_j$ is the stationary standard 
	deviation of the wealth distribution near attractor $j$, 
	which depends on $\sigma$, $f''_j(\kss_j)$, and the local 
	drift structure.\footnote{A local Gaussian approximation 
		near each attractor gives 
		$\sigma^{\mathrm{ss}}_j \approx \sigma / 
		\sqrt{2|\mu'_j(\kss_j)|}$, where 
		$\mu'_j = f'_j - \delta - c'_j$ is the derivative of the 
		drift evaluated at the attractor. This approximation is 
		valid when the density is sharply peaked around $\kss_j$.} 
	In the deterministic benchmark, the separation is:
	\[
	\kss_H(0) - \kss_L(0) = \left(\frac{\alpha}
	{\rho+\delta}\right)^{1/(1-\alpha)} 
	\bigl(A_H^{1/(1-\alpha)} - A_L^{1/(1-\alpha)}\bigr),
	\]
	which depends only on the TFP gap $A_H/A_L$ and the 
	structural parameters $(\alpha, \rho, \delta)$. In the 
	full stochastic model, regime switching and diffusion 
	shift both attractors 
	(Propositions~\ref{prop:L_shift} 
	and~\ref{prop:H_shift}). The net effect on the gap is 
	parameter-dependent: under the baseline calibration, 
	the gap widens from $4.00$ to $4.62$ because the 
	precautionary motive shifts $\kssc{H}$ upward by more 
	than it shifts $\kssc{L}$. In other regions of the 
	parameter space---particularly when $\lambda_{HL}$ is 
	large, strengthening the front-loading channel for 
	H-agents---the gap may narrow. Condition (C2) requires 
	that, whatever the direction of the shift, the residual 
	gap remain large relative to the diffusion-induced 
	spread. As $\sigma \to 0$, the standard deviations 
	$\sigma^{\mathrm{ss}}_j \to 0$ while the attractor gap 
	remains bounded away from zero, so (C2) is satisfied 
	for sufficiently small noise.
		
		\item \textbf{Interior signaling threshold.} The Skiba 
		threshold must be interior: $k^* \in (\phi, \infty)$, so 
		that signaling is feasible but not immediate 
		(Remark~\ref{rmk:skiba_cases}(a)). This requires two 
		sub-conditions:
		\begin{enumerate}[label=(\alph*)]
			\item $\phi < \kssc{L}$: the signaling cost is below the 
			L-attractor, so that an L-agent who receives an 
			opportunity can, in principle, accumulate enough to 
			signal without first needing to escape the basin.
			\item $D(\phi) > 0$: signaling is not optimal immediately 
			upon becoming affordable. This holds when the 
			obsolescence risk $\lambda_{HL}$ is large enough, or 
			the TFP gap $A_H/A_L$ small enough, that the post-signaling 
			wealth $k - \phi$ is too low to justify the transition 
			(Definition~\ref{def:surplus}).
		\end{enumerate}
		Together, (a) and (b) ensure that W-agents accumulate in the 
		interval $[\phi, k^*)$ before signaling---the ``Frustrated 
		Aspirants'' who generate mass between the two modes. If 
		$k^* = \phi$ (immediate signaling), the W-state is 
		transient and the economy behaves as a two-state model 
		without an accumulation phase. If $k^* = +\infty$ (no 
		signaling), the economy collapses to a single-attractor 
		model at $A_L$.
		
	\item \textbf{Moderate coupling rates.} The Poisson rates 
	$\lambda_{LH}$ and $\lambda_{HL}$ must be small enough 
	that agents spend sufficient time in each regime for the 
	local attractors to shape the distribution. Specifically:
	\begin{enumerate}[label=(\alph*)]
		\item The attractor gap $\kssc{H} - \kssc{L}$ 
		remains large relative to the diffusion-induced 
		spread 
		$\sigma^{\mathrm{ss}}_L + \sigma^{\mathrm{ss}}_H$. 
		Since the coupling shifts both attractors 
		(Propositions~\ref{prop:L_shift} 
		and~\ref{prop:H_shift}), this requires that the 
		net effect of the shifts---whether compressive or 
		expansive---does not reduce the gap below the 
		threshold needed for condition (C2).
		\item The obsolescence rate satisfies 
		$\lambda_{HL} \ll 1/T_{\mathrm{acc}}^{H}$, where 
		$T_{\mathrm{acc}}^{H}$ is the characteristic time 
		for an H-agent to accumulate from the 
		post-signaling entry point $k^* - \phi$ to 
		$\kssc{H}$. This ensures that newly promoted 
		H-agents have enough time to reach the upper 
		attractor before being hit by an obsolescence 
		shock. If $\lambda_{HL}$ is too large, H-agents 
		revert to L before accumulating to $\kssc{H}$, 
		and the right mode fails to form.
		\item The opportunity arrival rate satisfies 
		$\lambda_{LH} \ll 1/T_{\mathrm{acc}}^{W}$, where 
		$T_{\mathrm{acc}}^{W}$ is the time for a W-agent 
		to accumulate from $\kssc{L}$ to $k^*$. Combined 
		with (b), this ensures that agents do not cycle 
		rapidly between regimes: each agent spends enough 
		time near her current attractor for the local 
		restoring force to concentrate mass into a 
		distinguishable mode. When both rates are large, 
		agents cycle so frequently that the distribution 
		homogenizes into a single peak.
	\end{enumerate}
	\end{enumerate}
\end{proposition}

\begin{proof}[Proof sketch]
	Condition (C1) is established in 
	Propositions~\ref{prop:det_ss}--\ref{prop:stoch_ss}. The strict 
	concavity of $f_j$ ensures that the drift 
	$\mu_j(k) = f_j(k) - c_j(k) - \delta k$ changes sign exactly 
	once (from positive to negative) as $k$ increases through 
	$\kss_j$, making $\kss_j$ a unique, stable attractor.
	
	Condition (C2) is a separation condition on the coupled 
	attractors. When it holds, the stationary KFE 
	(Section~\ref{sec:kfe}) admits a solution with two local 
	maxima. To see this, note that near each attractor $\kss_j$, 
	the drift is approximately linear: 
	$\mu_j(k) \approx \mu'_j(\kss_j)(k - \kss_j)$ with 
	$\mu'_j(\kss_j) < 0$ (stability). The local density is 
	approximately Gaussian with mean $\kss_j$ and variance 
	$\sigma^2 / (2|\mu'_j(\kss_j)|)$. When the distance between 
	means far exceeds the sum of standard deviations, the two 
	Gaussians have negligible overlap and the aggregate density is 
	bimodal.
	
	Condition (C3) ensures that the signaling mechanism is active: 
	there exist agents in the W-state who accumulate toward $k^*$, 
	creating a flow from the lower basin to the upper basin. 
	Without this flow (i.e., when $k^* = \infty$), no agent reaches 
	the H-regime and the distribution is unimodal at $\kss_L$.
	
	Condition (C4) ensures that the flow between basins is slow 
	enough for two distinct modes to emerge. Rapid cycling 
	($\lambda_{LH}, \lambda_{HL}$ large) generates frequent 
	regime changes, preventing agents from settling near either 
	attractor and producing a broad, unimodal distribution.
\end{proof}

\begin{remark}[Sufficient vs.\ necessary conditions]%
	\label{rmk:sufficiency}
	Conditions (C1)--(C4) are jointly sufficient for bimodality 
	in the following sense: when all four hold, the numerical 
	solution of the KFE (Section~\ref{sec:kfe}) produces a density 
	with two distinct local maxima. The conditions are not 
	individually necessary---for example, bimodality can persist 
	even when (C4c) is mildly violated, because Brownian diffusion 
	provides an alternative (slower) channel for basin crossing. A 
	sharp characterization of the necessary and sufficient 
	conditions for bimodality in the full nonlinear KFE is an open 
	problem; the conditions above identify the economically 
	interpretable parameter restrictions under which bimodality 
	obtains.
\end{remark}

\begin{remark}[Numerical verification]\label{rmk:bimodal_numerical}
	Under the baseline calibration 
	(Table~\ref{tab:calibration}), all four conditions are 
	satisfied. Section~\ref{sec:kfe} presents the 
	resulting bimodal density.
\end{remark}

	%% %==========================================================%==
	\section{Numerical Method}\label{sec:numerical}
	%% %==========================================================%==
	
	The model does not admit a closed-form solution. The coupled 
	system of HJB equations \eqref{eq:hjb_L}--\eqref{eq:hjb_W}, 
	together with the variational inequality for the W-state and 
	the Kolmogorov Forward Equation for the stationary distribution, 
	must be solved numerically. This section describes the 
	computational strategy.
	
	The solution proceeds in two stages that exploit a fundamental 
	duality in the HACT framework \citep{Achdou2022}. The HJB 
	equations govern individual optimization: given the state 
	$(k, s)$, what is the optimal consumption and signaling policy? 
	The Kolmogorov Forward Equation (KFE) governs the distributional 
	consequences of those policies: given that all agents behave 
	optimally, what is the stationary cross-sectional distribution 
	of wealth? In the standard HACT literature with production 
	economies \citep{Achdou2022, Kaplan2018}, the two stages are 
	linked by equilibrium prices---the interest rate and the wage 
	depend on aggregate capital, creating a fixed-point problem. In 
	our model, the coupling is one-directional: since each agent 
	operates her own Cobb-Douglas technology without factor market 
	interaction, the individual problem does not depend on the 
	wealth distribution. The HJB system can therefore be solved 
	first, and the resulting policies fed into the KFE as given 
	inputs.
	
	\paragraph{Stage 1: HJB system.} We discretize the capital 
	domain $[0, k_{\max}]$ on a uniform grid of $N$ points and 
	solve the coupled HJB equations using the implicit upwind 
	finite difference scheme of \cite{Achdou2022}. The upwind 
	strategy---forward differences where the drift is positive, 
	backward differences where it is negative---ensures monotonicity 
	of the scheme, which is the key condition for convergence to the 
	viscosity solution \citep{barles1991}. The implicit 
	time stepping guarantees unconditional stability, allowing large 
	step sizes that accelerate convergence. Policy iteration 
	(Howard improvement) replaces the slow contraction of value 
	iteration with a sequence of linear solves, each of which 
	updates the value function globally. For the W-state, we 
	augment the standard scheme with an American projection step 
	that enforces the constraint 
	$\VW(k) \geq \VH(k - \phi)$ at each iteration, following the 
	penalty and projection methods used in computational finance for 
	American option pricing \citep{wilmott1995}.
	
	\paragraph{Stage 2: Kolmogorov Forward Equation.} Once the 
	optimal policies $\{c_j(k)\}_{j \in \{L,W,H\}}$ and the signal 
	region $\{k \geq k^*\}$ are determined, the stationary 
	distribution $\{g_L(k), g_W(k), g_H(k)\}$ is obtained from 
	the KFE---the adjoint of the infinitesimal generator associated 
	with the HJB operator. The KFE describes the steady-state 
	balance of probability flows: at each point $(k, j)$, the 
	inflow of agents from drift, diffusion, and regime switching 
	equals the outflow. The discretized KFE inherits the 
	transpose structure of the discretized HJB 
	\citep{Achdou2022}, ensuring consistency between the 
	individual optimization and the distributional accounting. The 
	non-local transfer induced by signaling---agents in W at 
	$k \geq k^*$ are removed and reinjected in H at 
	$k - \phi$---enters the KFE as a source-sink term that 
	connects the two basins of attraction.
	
	The remainder of this section presents the discretization, the 
	iterative algorithm, and the convergence diagnostics in detail.
%% ============================================================
\section{Numerical Method for the HJB System}\label{sec:numerical_hjb}
%% ============================================================

\subsection{Upwind Finite Differences}\label{sec:upwind}

We discretize $k \in [k_{\min}, k_{\max}]$ on a uniform grid 
$\{k_1, \ldots, k_N\}$ with $k_1 = \varepsilon > 0$ (small, to 
avoid the singularity of $k^{\alpha-1}$ at zero) and spacing 
$\Delta k = (k_{\max} - k_{\min})/(N-1)$.

\begin{definition}[Upwind approximation]\label{def:upwind}
	For each grid point $i$ and state $j$ with production 
	$f_j(k_i) = A_j k_i^\alpha$, define forward and backward 
	approximations to the first derivative:
	\begin{equation}\label{eq:fd_approx}
		V_{j,i}^{f} = \frac{V_{j,i+1} - V_{j,i}}{\Delta k}, 
		\qquad 
		V_{j,i}^{b} = \frac{V_{j,i} - V_{j,i-1}}{\Delta k}.
	\end{equation}
	Each approximation implies a candidate consumption policy 
	(from the first-order condition \eqref{eq:foc_crra}) and a 
	candidate drift:
	\begin{align}
		c_i^f = (V_{j,i}^f)^{-1/\gamma}, &\qquad 
		\mu_i^f = f_j(k_i) - c_i^f - \delta k_i, 
		\label{eq:cf_muf}\\
		c_i^b = (V_{j,i}^b)^{-1/\gamma}, &\qquad 
		\mu_i^b = f_j(k_i) - c_i^b - \delta k_i. 
		\label{eq:cb_mub}
	\end{align}
	The upwind scheme selects the approximation consistent with 
	the direction of capital flow:
	\begin{equation}\label{eq:upwind}
		V_{j,i}' = \begin{cases}
			V_{j,i}^f & \text{if } \mu_i^f > 0 
			\quad \text{(accumulating: use forward difference)}, \\[4pt]
			V_{j,i}^b & \text{if } \mu_i^b < 0 
			\quad \text{(decumulating: use backward difference)}, \\[4pt]
			u'(f_j(k_i) - \delta k_i) & \text{otherwise 
				\quad (zero drift: steady-state consumption)}.
		\end{cases}
	\end{equation}
	The third case corresponds to points where neither forward 
	nor backward drift has the correct sign. At such points, the 
	agent is locally at a zero-drift steady state: consumption 
	equals net income, $c_i = f_j(k_i) - \delta k_i$, and the 
	slope of the value function is $V_{j,i}' = u'(c_i)$.
\end{definition}

The discretized HJB operator for state $j$ at grid point $i$ 
takes the form:
\begin{equation}\label{eq:discrete_hjb}
	\rho V_{j,i} = u(c_{j,i}) + x_i V_{j,i+1} 
	+ z_i V_{j,i} + y_i V_{j,i-1} 
	+ \text{(switching terms)},
\end{equation}
where the finite-difference coefficients encode both advection 
(drift) and diffusion:
\begin{align}
	x_i &= \frac{\max(\mu_i, 0)}{\Delta k} 
	+ \frac{\sigma^2/2}{(\Delta k)^2}, \label{eq:xi}\\
	y_i &= \frac{-\min(\mu_i, 0)}{\Delta k} 
	+ \frac{\sigma^2/2}{(\Delta k)^2}, \label{eq:yi}\\
	z_i &= -(x_i + y_i). \label{eq:zi}
\end{align}
The coefficient $x_i$ governs the weight on the right neighbor 
$V_{j,i+1}$ (forward advection plus diffusion), $y_i$ governs 
the weight on the left neighbor $V_{j,i-1}$ (backward advection 
plus diffusion), and $z_i$ is the diagonal entry, chosen so that 
the row sums to zero in the generator matrix.

\begin{proposition}[Monotonicity and convergence]%
	\label{prop:monotone}
	The upwind scheme \eqref{eq:xi}--\eqref{eq:zi} produces a 
	system $A_j \mathbf{V}_j = \mathbf{b}_j$ where $A_j$ is a 
	diagonally dominant M-matrix (positive diagonal, non-positive 
	off-diagonal). The scheme is monotone, consistent, and 
	converges to the unique viscosity solution of the HJB equation 
	as $\Delta k \to 0$.
\end{proposition}

\begin{proof}[Proof sketch]
	By construction, $x_i \geq 0$ and $y_i \geq 0$ for all $i$ 
	(both the advection and diffusion contributions are 
	non-negative). Since $z_i = -(x_i + y_i) \leq 0$, the 
	off-diagonal entries of the generator matrix are non-positive. 
	Incorporating the discount rate, the diagonal of $A_j$ is 
	$\rho - z_i + \lambda_{\mathrm{out},j} = \rho + x_i + y_i 
	+ \lambda_{\mathrm{out},j} > 0$, ensuring strict diagonal 
	dominance. These properties make $A_j$ an M-matrix, which guarantees that the 
	discrete maximum principle holds: the solution is bounded by 
	the data. Monotonicity (the scheme is non-decreasing in each 
	$V_{j,i'}$ for $i' \neq i$) follows from the non-positivity 
	of off-diagonal entries. Combined with consistency (the 
	truncation error vanishes as $\Delta k \to 0$) and stability 
	(the M-matrix property provides uniform bounds), the 
	convergence theorem of \cite{barles1991} guarantees 
	convergence to the viscosity solution of the continuous HJB 
	equation.
\end{proof}
\subsection{Solution Algorithm}\label{sec:policy_iteration}

The discretized HJB system is nonlinear: the coefficients 
$x_i$, $y_i$, $z_i$ depend on the consumption policy 
$c_{j,i}$, which in turn depends on the value function 
$V_{j,i}$ through the first-order condition. We solve this 
fixed-point problem by an implicit time-stepping scheme 
with Gauss--Seidel ordering across regimes and an American 
projection step for the variational inequality in state W.

The algorithm iterates on the three value functions in the 
order $L \to W \to H$, which follows the direction of 
regime transitions. At each outer iteration $n$:

\begin{enumerate}[label=\textit{Step \arabic*.}]
	\item \textit{Policy update.} Given the current 
	iterate $\mathbf{V}_j^{n-1}$, compute the upwind 
	approximation \eqref{eq:upwind} at each grid point 
	to obtain the consumption policy $c_{j,i}^{n-1}$ and 
	drift $\mu_{j,i}^{n-1}$. Assemble the 
	finite-difference coefficients $x_i$, $y_i$, $z_i$ 
	from \eqref{eq:xi}--\eqref{eq:zi}.
	
	\item \textit{Implicit update.} For each state $j$ 
	in the order $L, W, H$, solve the tridiagonal system:
	\begin{equation}\label{eq:implicit_step}
		\left(\frac{1}{\Delta t} I + \rho I 
		+ \Lambda_j - A_j^{n-1}\right) 
		\mathbf{V}_j^{n} = 
		\frac{\mathbf{V}_j^{n-1}}{\Delta t} 
		+ \mathbf{u}^{n-1}_j 
		+ \Lambda_j\, \mathbf{V}_{-j}^{*},
	\end{equation}
	where $\mathbf{V}_{-j}^{*}$ uses the 
	\textit{most recently updated} value functions 
	of the other regimes (Gauss--Seidel ordering): 
	when updating W, $\mathbf{V}_L^n$ is already 
	available from the L-step; when updating H, 
	both $\mathbf{V}_L^n$ and $\mathbf{V}_W^n$ are 
	available. Each system is tridiagonal and solved 
	in $O(N)$ operations.
	
	\item \textit{American projection (W-state only).} 
	After solving for $\mathbf{V}_W^n$, enforce the 
	variational inequality: for each grid point 
	$k_i \geq \phi$, set
	\[
	V_{W,i}^{n} \leftarrow 
	\max\bigl(V_{W,i}^{n},\; 
	V_{H}^{n-1}(k_i - \phi)\bigr).
	\]
	Points where the constraint binds constitute the 
	signal region; the Skiba threshold $k^*$ is the 
	lowest such point. The consumption policy and 
	drift in the signal region are then recomputed 
	from the projected value function to ensure 
	consistency with the KFE.
	
	\item \textit{Convergence check.} Stop if 
	$\max_j \|\mathbf{V}_j^{n} 
	- \mathbf{V}_j^{n-1}\|_\infty 
	< \varepsilon_{\mathrm{tol}}$.
\end{enumerate}

The implicit scheme is unconditionally stable, 
allowing arbitrarily large time steps. Following 
\cite{Achdou2022}, we increase $\Delta t$ 
geometrically as the iteration converges: in the 
limit $\Delta t \to \infty$, each implicit step 
becomes an exact policy evaluation, and the algorithm 
reduces to Howard's policy iteration. Under the 
baseline calibration, convergence to 
$\varepsilon_{\mathrm{tol}} = 10^{-8}$ is achieved in 
8 iterations. Appendix~\ref{sec:gs} provides the 
full algorithmic details and convergence analysis.%% 
\section{Kolmogorov Forward Equation}\label{sec:kfe}
%% ============================================================

The stationary wealth distribution is the solution to the 
Kolmogorov Forward Equation (KFE)---the adjoint of the HJB 
operator---subject to the non-local transfer induced by 
signaling. For the sake of brevity, Appendix \ref{app:kfe}  derives the discretized KFE system, 
describes the signaling transfer mechanism, and characterizes 
the properties of the resulting ergodic distribution.

% ================================================
 \section{Calibration and Numerical Experiment}\label{sec:calibration}
 %% ============================================================
 
 The purpose of this section is not to match a specific economy 
 but to demonstrate that the model generates a bimodal stationary 
 distribution---the Twin Peaks---under empirically plausible 
 parameter values. We adopt a baseline calibration that places the 
 economy in the interior-threshold regime 
 (Remark~\ref{rmk:skiba_cases}(a)), where signaling is feasible 
 but not immediate, producing a non-trivial wait zone and the 
 full range of agent types described in the introduction.
 Table~\ref{tab:calibration} reports the baseline parameters. 
 We group them into four categories: preferences, technology, 
 stochastic environment, and signaling. For each parameter, we 
 adopt standard values from the literature where available and 
 discuss our choices in detail in Appendix \ref{app:param_discuss}.
 
 \begin{table}[htbp]
 	\centering
 	\caption{Baseline calibration}
 	\label{tab:calibration}
 	\begin{tabular}{llcl}
 		\toprule
 		Parameter & Description & Value 
 		& Reference / Target \\
 		\midrule
 		\multicolumn{4}{l}{\textit{Preferences}} \\[2pt]
 		$\gamma$ & Risk aversion (CRRA) & 2.0 
 		& Standard; \cite{LucasBook1987}, 
 		\cite{Attanasio1999}, \cite{Chetty2006} \\
 		$\rho$ & Discount rate & 0.05 
 		& Annual; \cite{Achdou2022} \\[6pt]
 		\multicolumn{4}{l}{\textit{Technology}} \\[2pt]
 		$\alpha$ & Capital elasticity & 0.33 
 		& Standard in Cobb-Douglas calibrations \\
 		$\delta$ & Depreciation rate & 0.02 
 		& Low; broad capital interpretation \\
 		$A_L$ & Low-regime TFP & 1.0 
 		& Normalization \\
 		$A_H$ & High-regime TFP & 1.25 
 		& Conservative; college premium 
 		$\approx 60$--$80\%$ 
 		(\cite{Goldin2008}) \\[6pt]
 		\multicolumn{4}{l}{\textit{Stochastic environment}} 
 		\\[2pt]
 		$\sigma$ & Diffusion coefficient & 0.30 
 		& Calibrated to match wealth dispersion \\
 		$\lambda_{LH}$ & Opportunity arrival rate & 0.005 
 		& Schumpeterian long waves; 
 		\cite{Comin2010} \\
 		$\lambda_{HL}$ & Obsolescence rate & 0.002 
 		& Persistent regimes; 
 		\cite{Azariadis2005} \\[6pt]
 		\multicolumn{4}{l}{\textit{Signaling}} \\[2pt]
 		$\phi$ & Signaling cost & 9.0 
 		& Interior Skiba; $\phi < \kssc{L} < k^*$ \\[4pt]
 		\midrule
 		\multicolumn{4}{l}{\textit{Computational}} \\[2pt]
 		$N$ & Grid points & 501 & \\
 		$k_{\max}$ & Upper grid boundary & 50 & \\
 		$\Delta t$ & Implicit time step & 500 & \\
 		\texttt{tol} & Convergence tolerance 
 		& $10^{-8}$ & \\
 		$\bar{\lambda}$ & Signaling drain rate (KFE) 
 		& $10^3$ & \\
 		\bottomrule
 	\end{tabular}
 \end{table}

 The next section presents the results of this experiment: 
 the value functions, optimal policies, signaling surplus, 
 and the stationary distribution.
 
 %% ============================================================
 \section{Results}\label{sec:results}
 This section presents the numerical results under the baseline 
 calibration of Table~\ref{tab:calibration}. We first build 
 intuition for the equilibrium structure, then verify the 
 theoretical predictions and examine the consumption diagnostics 
 that constitute the paper's central empirical prediction.

\subsection{Results in a Nutshell}\label{sec:nutshell}

Before examining the numerical output in detail, it is useful 
to build intuition for the equilibrium structure through three 
complementary perspectives: the spatial organization of the 
state space, the value functions that govern individual 
decisions, and the resulting stationary densities. 
Figures~\ref{fig:mechanism}--\ref{fig:state_densities} present 
these perspectives under the baseline calibration.

\smallskip\noindent\textit{The geography of the trap.}\enspace 
Figure~\ref{fig:mechanism} displays the partition of the 
capital axis induced by the signaling friction. The economy is 
divided into three regions whose boundaries are determined by 
two structural parameters---the signaling cost $\phi$ and the 
endogenous Skiba threshold $k^*$---and whose economic logic is 
fundamentally different.

Below $\phi = 9$, agents are too poor to afford the credential: 
even if a signaling opportunity arrives ($\lambda_{LH}$), they 
cannot exercise it. This is the region of \textit{absolute} 
exclusion, where the poverty trap operates through a wealth 
constraint that is, in a precise sense, analogous to a borrowing 
constraint in the Bewley tradition---except that here the 
constraint binds on the \textit{investment} margin rather than 
on consumption. Between $\phi$ and $k^* = 13.4$ lies the 
\textit{wait zone}: agents can afford the cost but 
\textit{choose} not to pay it, because the post-signaling 
wealth $k - \phi$ would leave them too far below the 
H-attractor, exposed to obsolescence risk with insufficient 
buffer. This is the region of \textit{rational} exclusion---the 
novel mechanism of our model. Above $k^*$, agents signal 
immediately upon opportunity arrival.

The critical observation is the position of the L-attractor: 
$\kssc{L} = 11.5$ falls \textit{inside} the wait zone 
($\phi < \kssc{L} < k^*$). This means that the typical 
low-productivity agent is neither too poor to signal nor 
unwilling to signal in principle---she is stuck in a region 
where signaling is affordable but not yet optimal. The trap is 
not a corner solution but an interior one, and this is what 
makes it invisible to standard diagnostic tests.

\begin{figure}[H]
	\centering
	\includegraphics[width=0.85\textwidth]{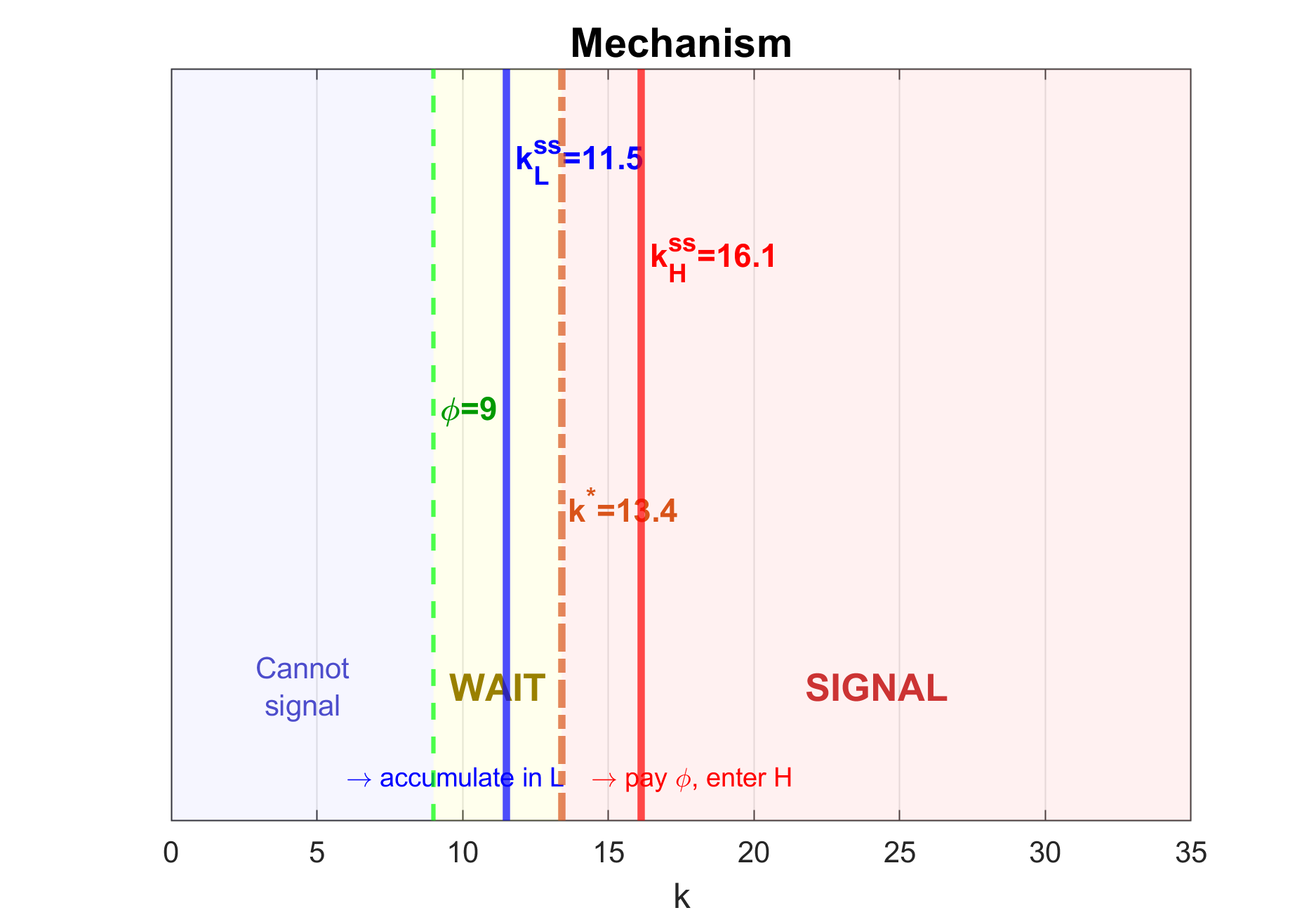}
	\caption{Partition of the capital axis under the baseline 
		calibration. The signaling cost $\phi = 9$ and the endogenous 
		Skiba threshold $k^* = 13.4$ divide the state space into three 
		regions: absolute exclusion ($k < \phi$), the wait zone 
		($\phi \leq k < k^*$, shaded in yellow), and immediate 
		signaling ($k \geq k^*$). The coupled attractors $\kssc{L} = 
		11.5$ and $\kssc{H} = 16.1$ are shown as solid vertical 
		lines. The L-attractor lies inside the wait zone: the 
		typical trapped agent \textit{can} afford to signal but 
		\textit{chooses} not to.}
	\label{fig:mechanism}
\end{figure}

\smallskip\noindent\textit{Why the trap is self-enforcing.}\enspace 
Figure~\ref{fig:value_functions} reveals the economic forces 
behind the partition. The three value functions $V_L$, $V_W$, 
$V_H$ are strictly concave (Proposition~\ref{prop:concavity}) 
and strictly ordered: $V_H(k) > V_W(k) > V_L(k)$ for all 
$k > 0$ (Proposition~\ref{prop:ordering}). The ordering 
reflects the value of options: $V_H$ incorporates the benefit 
of high productivity, $V_W$ adds the option to signal, and 
$V_L$ has neither. 

The signaling decision is governed by the comparison between 
$V_W(k)$---the value of holding the option---and 
$V_H(k - \phi)$---the value of exercising it. The dotted 
curve $V_H(k - \phi)$ is simply $V_H$ shifted rightward by 
$\phi$: paying the signaling cost is equivalent to starting 
the H-regime with $\phi$ fewer units of capital. At low wealth, 
this shift is devastating---the concavity of $V_H$ means that 
the marginal value of the lost capital is very high, so 
$V_W(k) \gg V_H(k - \phi)$ and the agent waits. As wealth 
increases, the gap narrows because the marginal cost of losing 
$\phi$ units falls (the function flattens). At $k = k^*$, the 
two curves cross: $V_W(k^*) = V_H(k^* - \phi)$, and the agent 
is indifferent---this is the Skiba threshold. 

The mechanism is self-enforcing because of the drift structure. 
At $\kssc{L}$, the drift $\mu_L$ is zero: the agent has no 
incentive to accumulate further under $A_L$ technology, so she 
\textit{cannot} drift into the signal region through savings 
alone. Only the exogenous arrival of an opportunity 
($\lambda_{LH}$) moves her to the W-state, and even then, 
she must accumulate from $\kssc{L} = 11.5$ to 
$k^* = 13.4$---an additional 1.9 units of capital---before 
signaling becomes optimal. During this accumulation phase, she 
is vulnerable to losing the opportunity before reaching $k^*$.

\begin{figure}[H]
	\centering
	\includegraphics[width=0.85\textwidth]{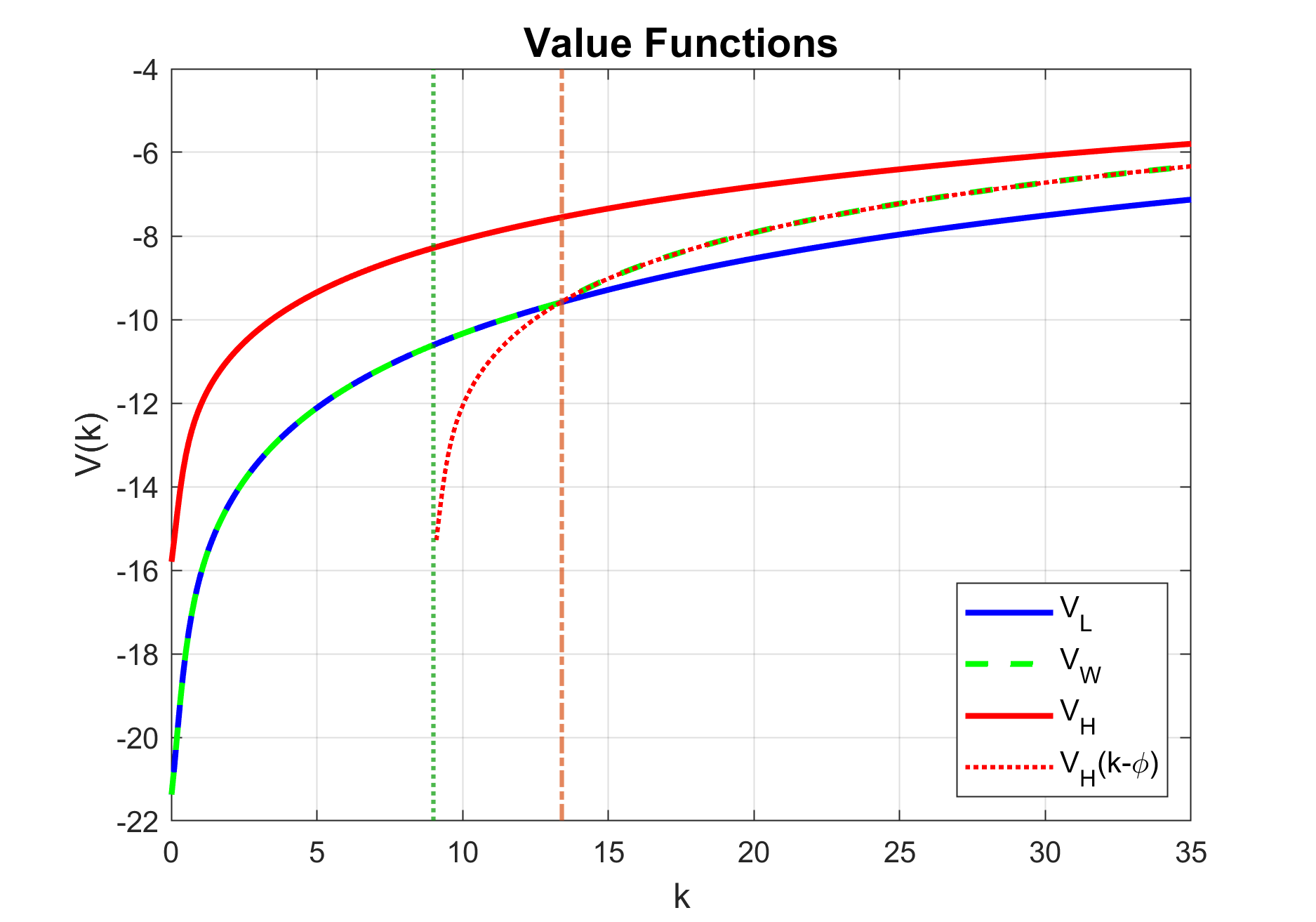}
	\caption{Value functions under the baseline calibration. 
		$V_H > V_W > V_L$ for all $k > 0$ 
		(Proposition~\ref{prop:ordering}). All three are strictly 
		concave (Proposition~\ref{prop:concavity}). The shifted 
		function $V_H(k - \phi)$ (dotted) crosses $V_W$ at the 
		Skiba threshold $k^* = 13.4$: below this point, the concavity 
		of $V_H$ makes the capital loss from paying $\phi$ too costly; 
		above it, the agent signals immediately. The green dotted line 
		marks $\phi = 9$, below which signaling is infeasible.}
	\label{fig:value_functions}
\end{figure}

\smallskip\noindent\textit{The resulting polarization.}\enspace 
Figure~\ref{fig:state_densities} displays the stationary 
densities decomposed by productivity regime---the observable 
counterpart of the theoretical structure. The bimodal shape 
predicted by Proposition~\ref{prop:bimodal} is immediately 
apparent, but the decomposition reveals \textit{why} the two 
peaks arise and why they persist.

The left mode, centered near $\kssc{L}$, is dominated by the 
L-density $g_L$ (blue). These are agents who have never received 
an opportunity, or who received one and lost it before reaching 
$k^*$. The concentration of mass near the attractor reflects 
the restoring force of the drift: agents who are perturbed away 
from $\kssc{L}$ by the diffusion $\sigma$ are pulled back by 
the negative drift above $\kssc{L}$ and positive drift below 
it. The right mode, centered near $\kssc{H}$, is dominated by 
the H-density $g_H$ (red). These are agents who have 
successfully signaled and are enjoying the returns to high 
productivity. The peak is sharper and taller, reflecting the 
stronger restoring force at $\kssc{H}$ (where the curvature of 
$f_H$ is larger due to higher TFP).

The W-density $g_W$ (green dashed) is small but structurally 
revealing. It is concentrated in the interval $[\phi, k^*]$---
precisely the wait zone of Figure~\ref{fig:mechanism}---and 
vanishes above $k^*$, where agents signal immediately. Its 
mass ($\pi_W = 11.8\%$) represents the population of 
Frustrated Aspirants: agents who hold a signaling opportunity 
but have not yet accumulated enough to exercise it. The 
existence of this intermediate group---neither trapped nor 
free---is a distinctive prediction of the model, absent from 
both the Bewley tradition (which lacks regime switching) and 
the Galor-Zeira tradition (which lacks the option-value margin).

A striking feature of the figure is that the three densities 
overlap extensively. This overlap is not noise---it is 
generated by the flows between regimes and constitutes one of 
the model's sharpest empirical implications: capital alone does 
not identify an agent's state; two households with identical 
wealth can be on opposite trajectories. Three overlap zones are 
visible, each with a distinct economic story.

In the region $k \in [4, \phi]$, the left tail of $g_H$ 
runs beneath $g_L$. The H-agents here are \textit{newly 
	signaled}: they entered state H at $k^* - \phi = 4.4$ and 
are accumulating rapidly toward $\kssc{H}$. The density is 
thin because the drift $\mu_H$ is strongly positive at low 
capital---agents transit through this zone quickly, like a 
fast current through a narrow channel. Their L-counterparts 
at the same capital level are in a fundamentally different 
situation: slowly accumulating toward $\kssc{L}$ under low 
productivity $A_L$, with much lower savings rates. An 
econometrician observing two households with $k = 6$ would 
find radically different consumption behavior---the H-agent 
saving aggressively (high marginal return $f'_H(6)$), the 
L-agent consuming most of her income---yet no cross-sectional 
variable other than the productivity regime would distinguish 
them.

In the wait zone $k \in [\phi, k^*]$, all three densities 
coexist. This is the most economically complex region of the 
state space: $g_L$ is approaching its peak (we are near 
$\kssc{L} = 11.5$); $g_W$ has its entire support here, 
representing Frustrated Aspirants who hold the signaling 
option and are saving toward $k^*$; and $g_H$ is still 
present as a transit flow of recently promoted agents 
climbing toward $\kssc{H}$. Three agents with $k = 11$ can 
be in entirely different situations: one trapped without an 
opportunity (L), one racing to accumulate the last 2.4 units 
before the opportunity expires (W), and one freshly 
credentialed and building her buffer (H). Their MPCs differ, 
their APCs differ, and their future trajectories diverge 
completely---yet in a household survey they are 
observationally identical without information on the 
productivity regime.

In the region $k \in [k^*, 20]$, the right tail of $g_L$ 
runs beneath the peak of $g_H$. The L-agents here are the 
\textit{Decaying Rentiers}---former H-agents who suffered 
the obsolescence shock ($\lambda_{HL}$) and reverted to low 
productivity. The shock changes the regime instantaneously 
but does not destroy capital: these agents retain high wealth 
but now face a negative drift under $A_L$ technology, which 
cannot sustain such a capital stock. They are decumulating 
toward $\kssc{L}$, consuming more than their net income. 
They are the mirror image of the newly signaled agents in 
the left tail of $g_H$: same wealth region, opposite 
direction, opposite regime. The symmetry between these two 
transit populations---one ascending in H, one descending in 
L---is a direct consequence of the Poisson switching 
structure and illustrates why regime identification is 
essential for any empirical test of the model.

The valley between the two peaks---the region around 
$k^*$---has near-zero density. This is not an accident: agents 
in the W-state who reach $k^*$ signal immediately, removing 
themselves from the W-density and re-entering as H-agents at 
$k^* - \phi$. The signaling threshold acts as an absorbing 
barrier for the wait state, creating a structural gap in the 
distribution that separates the two clubs. The depth of this 
valley---rather than the distance between the peaks---is the 
most robust measure of the strength of the poverty trap.

\begin{figure}[H]
	\centering
	\includegraphics[width=0.85\textwidth]{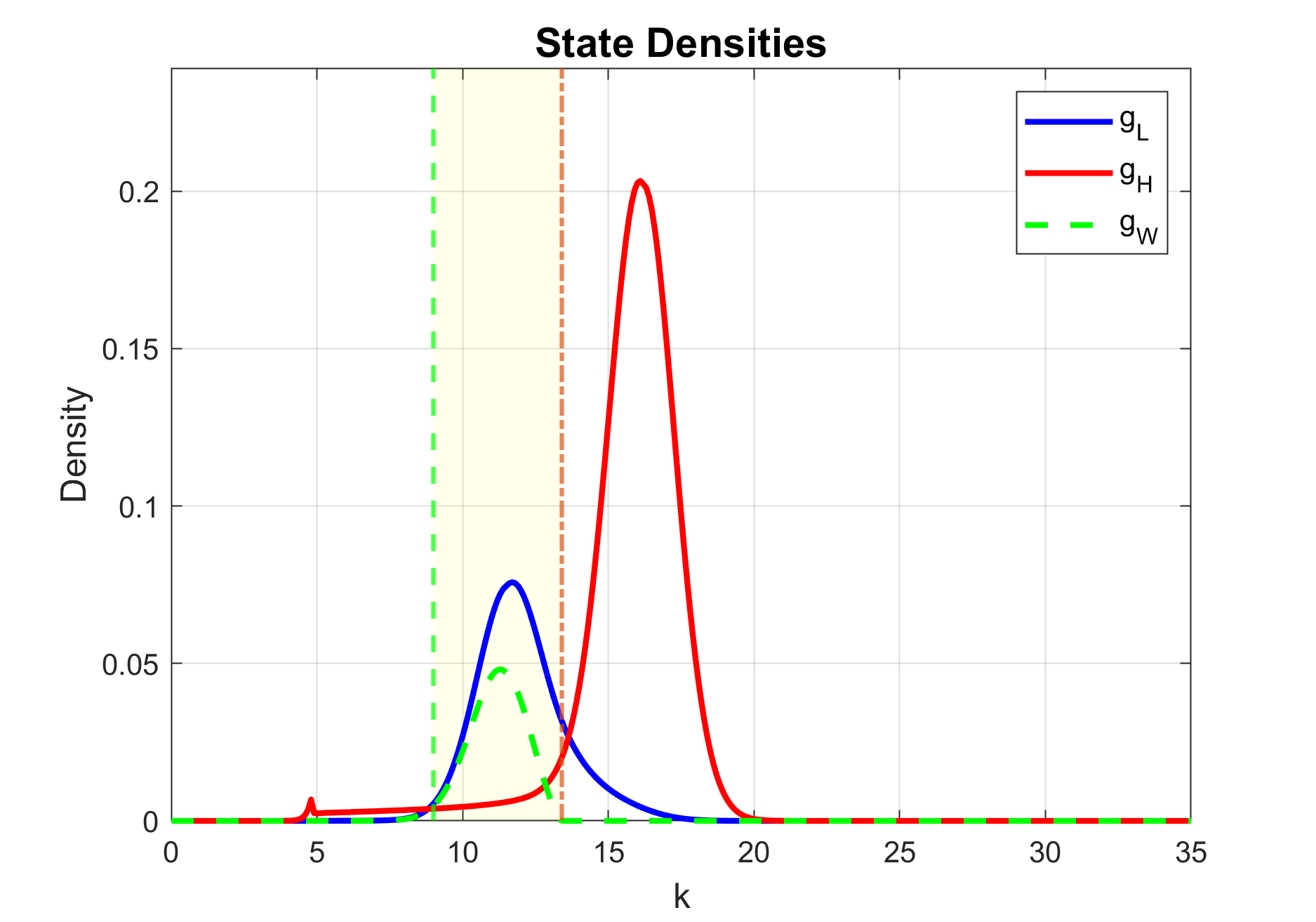}
	\caption{Stationary densities by regime. The L-density $g_L$ 
		(blue) peaks near $\kssc{L} = 11.5$, forming the poverty trap 
		mode; the H-density $g_H$ (red) peaks near $\kssc{H} = 16.1$, 
		forming the accumulation mode. The W-density $g_W$ (green 
		dashed) is confined to the wait zone $[\phi, k^*]$ and vanishes 
		above $k^*$, where signaling is immediate. The three densities 
		overlap extensively: in $[4, \phi]$, newly signaled H-agents 
		transit rapidly alongside slowly accumulating L-agents; in 
		$[\phi, k^*]$, all three populations coexist; in $[k^*, 20]$, 
		Decaying Rentiers (L) slide back through the region dominated 
		by H-agents. The valley near $k^*$ reflects the 
		absorbing-barrier effect: agents who reach the threshold exit 
		the W-state and re-enter as H-agents at $k^* - \phi$. 
		Ergodic shares: $\pi_L = 25.2\%$, $\pi_W = 11.8\%$, 
		$\pi_H = 63.0\%$.}
	\label{fig:state_densities}
\end{figure}

\subsection{Numerical Verification}\label{sec:verification}

We now verify that the simulation output is consistent with 
the theoretical predictions of 
Sections~\ref{sec:model}--\ref{sec:bimodality}.

\smallskip\noindent\textit{Convergence.}\enspace 
The HJB solver (Algorithm~\ref{alg:hjb}) converges in 8 
iterations to a tolerance of $10^{-8}$. The coupled 
attractors are $\kssc{L} = 11.50$ and $\kssc{H} = 16.12$, 
both above their deterministic counterparts 
($\kss_L(0) = 10.12$, $\kss_H(0) = 14.12$). The 
L-attractor rises by 13.6\%, driven by the precautionary 
motive and the option value of signaling opportunities 
(Propositions~\ref{prop:stoch_ss} and~\ref{prop:L_shift}). 
The H-attractor rises by 14.2\%, confirming that the 
precautionary channel dominates the front-loading effect 
(Proposition~\ref{prop:H_shift}). The attractor gap widens 
modestly from $4.00$ to $4.62$.

\smallskip\noindent\textit{Skiba threshold.}\enspace 
The threshold is $k^* = 13.40$, producing an interior 
signaling regime (Remark~\ref{rmk:skiba_cases}(a)) with a 
wait zone of width $k^* - \phi = 4.40$. The signaling 
surplus $D(k)$ declines monotonically from $D(\phi) = 5.25$ 
to $D(k^*) = 0$, and the L-attractor lies inside the wait 
zone ($\phi < \kssc{L} < k^*$), confirming the structural 
trap mechanism.

\smallskip\noindent\textit{Bimodality.}\enspace 
The stationary distribution is bimodal, with ergodic shares 
$\pi_L = 25.2\%$, $\pi_W = 11.8\%$, $\pi_H = 63.0\%$. 
Mean wealth is $\E[k] = 14.23$ and the Gini coefficient is 
$0.104$. The modest Gini reflects the conservative TFP gap 
($A_H/A_L = 1.25$); the relevant diagnostic of the poverty 
trap is the bimodality of the distribution, not the level 
of inequality.

\smallskip\noindent\textit{Euler equation.}\enspace 
At both attractors, the drift is numerically zero 
($|\mu_j(\kssc{j})| < 10^{-6}$), confirming the zero-drift 
condition. The net marginal return 
$f'_j(\kssc{j}) - \delta = 0.044$ falls below 
$\rho = 0.050$ by $0.006$, representing the combined 
precautionary and switching corrections in the Euler 
equation~\eqref{eq:euler_derived}. The near-equality of net 
returns across regimes ($f'_L - \delta \approx f'_H - \delta 
\approx 0.044$) reflects the symmetric structure of the 
zero-drift condition: at both attractors, the risk-adjusted 
return, net of all corrections, equals $\rho$.

\subsection{The MPC/APC Diagnostic}\label{sec:mpc_results}

The central empirical prediction concerns the joint behavior 
of the marginal propensity to consume out of wealth 
($\partial c / \partial k$) and the average propensity to 
consume ($c/Y$). Table~\ref{tab:mpc_apc} reports these 
statistics at the two attractors.

\begin{table}[htbp]
	\centering
	\caption{MPC and APC at the attractors}
	\label{tab:mpc_apc}
	\begin{tabular}{lcccc}
		\toprule
		& $\partial c_j/\partial k$ & $c_j/Y_j$ 
		& $f'_j - \delta$ & Drift $\mu_j$ \\
		\midrule
		L-attractor ($k = 11.50$) 
		& 0.055 & 0.897 & 0.044 & $\approx 0$ \\
		H-attractor ($k = 16.12$) 
		& 0.073 & 0.897 & 0.044 & $\approx 0$ \\
		\midrule
		Benchmark $\rho$ & 0.050 & --- & 0.050 & --- \\
		\bottomrule
	\end{tabular}
\end{table}

At $\kssc{L}$, the MPC is $0.055 \approx \rho$: a windfall 
$\Delta k$ raises consumption by only 
$\rho \cdot \Delta k$ per unit time. The agent has no 
incentive to save the transfer---the net marginal return is 
approximately $\rho$, just enough to compensate impatience 
---so it is absorbed into consumption while capital remains 
practically unchanged at $\kssc{L}$. Meanwhile, her APC is 
$0.897$: she consumes nearly all her income, saving nothing 
in net terms.

This combination---low MPC, high APC---is the diagnostic 
signature of the structural trap. Under the standard 
Bewley-Kaplan-Violante framework, low MPC signals financial 
comfort and high APC signals financial stress. In our model, 
both are consequences of the zero-drift condition at a 
low-productivity attractor. The agent is not ``hitting a 
wall'' (a binding constraint, which would produce 
\textit{high} MPC) but ``sitting in a hole'' (a 
low-productivity basin, which produces low MPC and high APC 
simultaneously).

The population-weighted MPCs confirm the cross-sectional 
pattern:
\begin{align*}
	\E[\partial c/\partial k \mid L] &= 0.090, \qquad
	\E[\partial c/\partial k \mid W] = 0.085, \\
	\E[\partial c/\partial k \mid H] &= 0.097, \qquad
	\E[\partial c/\partial k] = 0.094.
\end{align*}
The L/H ratio is $0.93$: structurally trapped agents have 
\textit{lower} marginal responsiveness to wealth than 
high-productivity agents, despite being poorer. This is the 
opposite of what a borrowing constraint model predicts and 
provides a testable cross-sectional implication.

\subsection{Agent Phenotypes}\label{sec:phenotypes}

The simulation identifies five distinct phenotypes, 
characterized by their position in the state-capital space 
and their consumption diagnostics:

\begin{enumerate}[label=(\roman*)]
	\item \textit{Hand-to-mouth} ($k < 2$, state L): 
	negligible share ($< 0.01\%$). High MPC 
	($\approx 0.25$), low APC ($\approx 0.64$). The 
	closest analogue to the Bewley-Kaplan-Violante 
	constrained agent.
	\item \textit{Structurally trapped} 
	($\phi < k \approx \kssc{L}$, state L): share 
	$\approx 20\%$. Low MPC ($\approx \rho$), high APC 
	($\approx 0.90$). At the zero-drift interior optimum.
	\item \textit{Frustrated Aspirants} (state W, wait 
	zone): share $11.8\%$. MPC slightly below the 
	L-average because the option value raises the shadow 
	price of wealth.
	\item \textit{Decaying Rentiers} ($k > k^*$, state L): 
	share $4.2\%$. Former H-agents sliding back toward 
	$\kssc{L}$ after the obsolescence shock. Low MPC, 
	temporarily high capital.
	\item \textit{Successful signalers} (state H): share 
	$63.0\%$. MPC $\approx 0.097$, concentrated near 
	$\kssc{H}$.
\end{enumerate}

\section{Policy Implications}\label{sec:policy}

The model's central implication for policy is that the 
standard toolkit---wealth transfers and credit 
relaxation---is ineffective against structural traps. An 
agent at $\kssc{L}$ who receives a transfer $\Delta k$ 
has no incentive to save it: the net marginal return is 
approximately $\rho$, so the windfall is absorbed into 
consumption and capital drifts back to $\kssc{L}$. The 
transfer is palliative, not structural. Effective policy 
must instead alter the parameters that define the trap: 
the signaling cost, the opportunity arrival rate, or the 
obsolescence rate.

\subsection{Policy Channels}\label{sec:instruments}

The model identifies four distinct channels, each 
targeting a different source of immobility.

\smallskip\noindent\textit{Reducing the signaling cost 
	$\phi$.}\enspace This is the most powerful lever. A 
lower $\phi$ simultaneously makes signaling affordable 
at lower wealth and lowers the Skiba threshold $k^*$, 
shrinking the wait zone and increasing the flow from L 
to H. In the limit $\phi \to 0$, signaling is immediate 
upon opportunity arrival and the poverty trap vanishes. 
Concrete examples include tuition subsidies, credential 
recognition programs, and reduced barriers to 
professional licensing. The comparative statics confirm 
that as $\phi$ increases, the economy transitions through 
the three regimes of Remark~\ref{rmk:skiba_cases}: 
immediate signaling, interior threshold (the baseline), 
and no signaling ($k^* = \infty$, H-density identically 
zero).

\smallskip\noindent\textit{Increasing opportunity arrival 
	$\lambda_{LH}$.}\enspace A higher $\lambda_{LH}$ raises 
the option value for L-agents, pushing $\kssc{L}$ upward 
(Proposition~\ref{prop:L_shift}) and increasing the flow 
into the H-regime. This is the channel through which labor 
market institutions operate: job matching services, 
information networks, and credential transparency all 
increase the rate at which low-productivity agents 
encounter signaling opportunities. However, very high 
$\lambda_{LH}$ accelerates turnover between regimes, 
potentially merging the two modes into a single peak 
(condition C3 of Section~\ref{sec:bimodality}).

\smallskip\noindent\textit{Reducing obsolescence risk 
	$\lambda_{HL}$.}\enspace A lower $\lambda_{HL}$ yields a 
double dividend: it extends the expected duration of the 
H-regime, making signaling more attractive (lowering 
$k^*$), and it extends the window for W-agents to 
accumulate toward $k^*$ before losing the opportunity. 
Policies targeting continuous professional development, 
retraining, and adaptive credentialing operate through 
this channel.

\smallskip\noindent\textit{Raising the TFP gap 
	$A_H/A_L$.}\enspace A larger productivity gap widens the 
distance between attractors and raises the return to 
signaling, lowering $k^*$. Both effects increase $\pi_H$, 
but also deepen polarization: the average agent is richer, 
while the trapped agents are relatively worse off. In 
practice, this corresponds to investments in the 
\textit{quality} of human capital (raising the 
productivity premium of education) rather than its 
accessibility.

\subsection{Comparative Statics}\label{sec:compstatics}

We briefly discuss how the remaining structural parameters 
affect the equilibrium, focusing on the three objects that 
define the trap: $k^*$, the attractor gap 
$\kssc{H} - \kssc{L}$, and the ergodic shares.

\smallskip\noindent\textit{Capital elasticity 
	$\alpha$.}\enspace As $\alpha$ increases toward 1, 
diminishing returns weaken: the attractors move rightward 
and farther apart, but the basins become shallower, 
allowing the diffusion to spread each mode. In the limit 
$\alpha \to 1$, stable interior attractors vanish 
(Remark~\ref{rmk:cd_not_linear}) and the poverty trap 
disappears. A lower $\alpha$ strengthens the trap by 
sharpening diminishing returns and steepening the basins.

\smallskip\noindent\textit{Diffusion $\sigma$.}\enspace 
Larger $\sigma$ spreads each mode (the stationary standard 
deviation scales as 
$\sigma / \sqrt{2|\mu'_j(\kssc{j})|}$) and strengthens 
the precautionary motive, shifting both attractors upward 
(Proposition~\ref{prop:stoch_ss}). For bimodality, the 
attractor gap must exceed the combined spread (condition 
C2): sufficiently large $\sigma$ merges the two modes 
into one.

\subsection{Diagnostic Implications}\label{sec:lessons}

The MPC/APC taxonomy provides a practical tool for 
identifying the nature of deprivation and directing 
interventions accordingly. Table~\ref{tab:diagnostic} 
summarizes the mapping from observables to policy.

\begin{table}[htbp]
	\centering
	\caption{Diagnostic guide: from observables to policy}
	\label{tab:diagnostic}
	\begin{tabular}{lcccl}
		\toprule
		Phenotype & MPC & APC & Drift 
		& Intervention \\
		\midrule
		Liquidity constrained 
		& High & Low & $\mu > 0$ 
		& Wealth transfer \\
		Structurally trapped 
		& Low & High & $\mu = 0$ 
		& Reduce $\phi$ or raise $\lambda_{LH}$ \\
		Frustrated Aspirant 
		& Low & High & $\mu > 0$ 
		& Reduce $k^*$ or slow $\lambda_{HL}$ \\
		Decaying Rentier 
		& Low & High & $\mu < 0$ 
		& Slow $\lambda_{HL}$ \\
		\bottomrule
	\end{tabular}
\end{table}

Standard consumption-based tests detect only the first 
group (high MPC). The structurally trapped, Frustrated 
Aspirants, and Decaying Rentiers---who together account 
for over 35\% of the population---are invisible to these 
tests because their low MPC mimics the signature of an 
unconstrained optimizer. An empirical strategy that 
conditions on \textit{both} MPC and APC can separate the 
two populations and target interventions accordingly: 
transfers to the liquidity constrained, structural reforms 
for the trapped.
%% ============================================================
\section{Conclusion}\label{sec:conclusion}

This paper embeds a signaling friction into the 
continuous-time heterogeneous agent framework. Agents 
operate Cobb-Douglas technologies with regime-specific 
productivity $A_j \in \{A_L, A_H\}$ and face stochastic 
arrival of signaling opportunities and skill obsolescence 
risk. The optimal stopping problem---when to pay the 
lump-sum cost $\phi$ to upgrade productivity---generates 
an endogenous Skiba threshold $k^*$ that partitions the 
state space into three regions: absolute exclusion, 
rational exclusion (the wait zone), and immediate 
signaling. The stationary distribution exhibits Twin 
Peaks, with agents in three distinct states coexisting 
at the same wealth levels.

The paper's main contribution is to identify a 
structurally distinct form of poverty trap and its 
diagnostic signature. The trap is an interior optimum, 
not a corner solution: at the L-attractor, diminishing 
returns have driven the risk-adjusted marginal return to 
saving down to the discount rate, so no further 
accumulation is optimal. The agent is not constrained 
---she is rationally immobile. This produces a 
consumption signature that is invisible to standard 
tests: low MPC out of wealth ($\partial c/\partial k 
\approx \rho$) combined with high APC ($c/Y \approx 
0.90$). The MPC alone would classify the agent as 
unconstrained; the APC alone would suggest financial 
stress. Only the joint observation reveals the trap.

A second contribution is to show that capital alone is 
insufficient to identify an agent's position in the 
polarization dynamics. The decomposition of the 
stationary distribution by regime reveals that 
structurally trapped, waiting, and successfully 
upgraded agents coexist at the same wealth levels with 
radically different consumption behavior and mobility 
prospects. This has direct implications for empirical 
work: any test of the model requires information on the 
productivity regime, not only on wealth.

The model also delivers a behavioral taxonomy that maps 
naturally onto policy. The MPC/APC diagnostic separates 
agents who need transfers (liquidity constrained, high 
MPC) from those who need structural reform (trapped, low 
MPC). The most effective intervention targets the 
signaling cost $\phi$---through subsidized training, 
credential reform, or lower barriers to entry---rather 
than the wealth level. Slowing obsolescence 
($\lambda_{HL}$) yields a double dividend by 
simultaneously raising the value of the H-regime and 
extending the window for reaching $k^*$.

Several extensions are natural. First, decomposing 
$\lambda_{HL}$ into separate rates for opportunity loss 
(W $\to$ L) and obsolescence (H $\to$ L) would break the 
common-rate symmetry and allow for richer policy analysis. 
Second, embedding the model in an overlapping generations 
framework would connect the signaling decision to 
intergenerational mobility \citep{Chetty2014}. Third, 
introducing endogenous wages through a matching model 
would allow for general equilibrium effects on the 
productivity premium. Finally, the consumption diagnostic 
is testable with household panel data: the model predicts 
a cluster of households with simultaneously low MPC and 
high APC that standard tests would classify as 
unconstrained but that are, in fact, structurally 
immobile.
	
	%\nocite{*}
	\bibliography{biblio3}
	
\section{Appendix} \label{app}

\subsection{Proof of Proposition \ref{prop:concavity}} \label{app:A1}
\begin{proof}
	We prove the result for $\VL$; the argument for $\VH$ is 
	identical, and the result for $\VW$ follows by an additional 
	argument addressing the signaling option.
	
	\medskip\noindent\textit{Step 1: Concavity of $\VL$ and $\VH$.}
	Fix a regime $j \in \{L, H\}$ and consider two initial conditions 
	$k^a, k^b \in (0, \infty)$ with $k^a \neq k^b$. Let 
	$\{c^a_t\}$ and $\{c^b_t\}$ denote the respective optimal 
	consumption policies, and let $\{k^a_t\}$ and $\{k^b_t\}$ 
	denote the resulting capital paths. For $\theta \in (0,1)$, 
	define $k^\theta_0 = \theta k^a_0 + (1-\theta) k^b_0$ and 
	consider the \textit{feasible} (not necessarily optimal) policy 
	$c^\theta_t = \theta c^a_t + (1-\theta) c^b_t$ starting from 
	$k^\theta_0$.
	
	Since the state equation \eqref{eq:state} is linear in $(k, c)$ 
	for a given Brownian path $\{W_t\}$ (the drift 
	$f_j(k) - c - \delta k$ is concave in $k$ and linear in $c$, 
	and the diffusion $\sigma$ is constant), the capital path under 
	$c^\theta_t$ satisfies, pathwise:
	\[
	k^\theta_t \geq \theta k^a_t + (1-\theta) k^b_t 
	\quad \text{for all } t \geq 0,
	\]
	where the inequality follows from the concavity of 
	$f_j(k) = A_j k^\alpha$ in $k$ ($\alpha < 1$). By strict 
	concavity of the CRRA utility function $u$:
	\[
	u(c^\theta_t) = u(\theta c^a_t + (1-\theta) c^b_t) 
	> \theta\, u(c^a_t) + (1-\theta)\, u(c^b_t).
	\]
	Taking expectations and integrating:
	\[
	V_j(k^\theta_0) \geq \E\left[\int_0^\infty e^{-\rho t}\, 
	u(c^\theta_t)\, \dd t\right] 
	> \theta\, V_j(k^a_0) + (1-\theta)\, V_j(k^b_0),
	\]
	where the first inequality uses the fact that $c^\theta_t$ is 
	feasible (but not necessarily optimal) from $k^\theta_0$, and 
	the second uses the strict concavity of $u$. This establishes 
	strict concavity of $V_j$ on $(0, \infty)$.
	
	\medskip\noindent\textit{Step 2: Extension to $\VW$.}
	The W-agent has the same production function as the L-agent 
	($A_L k^\alpha$) plus an additional option: the right to signal 
	at a self-chosen time $\tau$. The argument of Step~1 applies 
	to the ``continuation value'' component of $\VW$ (the HJB part 
	of the variational inequality). It remains to show that the 
	signaling option preserves concavity.
	
	For $k \geq \phi$, we have 
	$\VW(k) = \max\bigl(\VW^{\text{cont}}(k),\, \VH(k - \phi)\bigr)$, 
	where $\VW^{\text{cont}}$ is the continuation value (the 
	solution to the unconstrained HJB). By Step~1, 
	$\VW^{\text{cont}}$ is strictly concave. By Step~1 applied to 
	regime $H$, $\VH$ is strictly concave, and hence 
	$k \mapsto \VH(k - \phi)$ is strictly concave on 
	$[\phi, \infty)$. The pointwise maximum of two concave 
	functions is concave (though not necessarily strictly so). 
	Strict concavity of $\VW$ follows because, at each $k$, exactly 
	one of the two functions determines $\VW$ (by the complementary 
	slackness structure of the variational inequality 
	\eqref{eq:hjb_W}), and both are strictly concave on their 
	respective active regions. At the boundary $k^*$, smooth 
	pasting \eqref{eq:smooth_pasting} ensures that $\VW$ is $C^1$, 
	so $\VW$ is strictly concave on $(0, \infty)$.
\end{proof}
\subsection{Proof of Proposition \ref{prop:ordering}}
\label{app:A2}
\begin{proof}
	
	\textit{Part (i): $\VH(k) > \VW(k) > \VL(k)$ for all 
		$k > 0$.}
	
	\smallskip\noindent
	We establish the two inequalities separately.
	
	\smallskip\noindent
	\textit{Step 1: $\VH(k) > \VL(k)$.} Fix $k > 0$ and 
	consider the H-agent mimicking the L-agent's optimal 
	consumption policy $\{c^L_t\}$. Under this strategy, 
	the H-agent's capital satisfies:
	\[
	\dd k^H_t = \bigl[A_H (k^H_t)^\alpha - c^L_t 
	- \delta k^H_t\bigr]\dd t + \sigma\, \dd W_t.
	\]
	Since $A_H > A_L$ and both agents start at $k_0 = k$, 
	the drift is strictly higher at every instant where 
	$k^H_t > 0$ (which holds almost surely by the 
	reflecting barrier). By the comparison theorem for 
	SDEs with common noise, $k^H_t \geq k^L_t$ pathwise, 
	and the mimicking strategy is therefore feasible. The 
	H-agent consuming $\{c^L_t\}$ obtains exactly $\VL(k)$. 
	Since mimicking is not optimal---the H-agent can 
	strictly improve by re-optimizing given $A_H > A_L$---we 
	conclude $\VH(k) > \VL(k)$.
	
	\smallskip\noindent
	\textit{Step 2: $\VW(k) > \VL(k)$.} The W-agent can 
	mimic the L-agent's consumption policy while ignoring 
	the signaling opportunity entirely. Under this strategy, 
	she produces with $A_L$ (since $A_W = A_L$) and follows 
	the same capital path as the L-agent. If the 
	opportunity is lost at rate $\lambda_{HL}$, she 
	transitions to state L and continues with the 
	L-optimal policy, which she is already following---so 
	the transition has no effect on her payoff. The 
	``ignore the option'' strategy therefore yields exactly 
	$\VL(k)$. Since $\VW$ is the supremum over all 
	strategies, $\VW(k) \geq \VL(k)$. The inequality is 
	strict for all $k > 0$ whenever $\kstar < \infty$ 
	(Remark~\ref{rmk:skiba_cases}, cases (a)--(b)): the 
	signaling option has positive value because there 
	exists a future wealth level at which exercise is 
	optimal, and the positive probability of reaching 
	that level (via drift and diffusion) gives the option 
	strictly positive present value.
	
	\smallskip\noindent
	\textit{Step 3: $\VH(k) > \VW(k)$.} The H-agent is 
	already at productivity $A_H$ without having to pay 
	any cost. The W-agent produces at $A_L$ and, to reach 
	$A_H$, must pay $\phi$ at some future date---if she 
	reaches $\kstar$ before the opportunity expires. 
	Formally, consider the H-agent mimicking the W-agent's 
	optimal strategy (same consumption, same signaling 
	timing if applicable). Since $A_H > A_L = A_W$, the 
	H-agent's capital path dominates the W-agent's by the 
	same comparison argument as Step~1. Moreover, the 
	H-agent need not pay any signaling cost and faces no 
	risk of losing a pending opportunity. She obtains at 
	least the same consumption stream at strictly lower 
	cost, so $\VH(k) \geq \VW(k)$. The inequality is 
	strict because the W-agent operates under $A_L$ for a 
	positive expected duration before (possibly) reaching 
	$A_H$, while the H-agent enjoys $A_H$ from $t = 0$.
	
	\medskip\noindent
	\textit{Part (ii): $\VW(k) \geq \VH(k - \phi)$ for 
		$k \geq \phi$, with equality if and only if 
		$k \geq \kstar$.}
	
	\smallskip\noindent
	The inequality follows immediately from the definition 
	of $\VW$: immediate signaling at $k \geq \phi$ is an 
	admissible strategy yielding payoff $\VH(k - \phi)$, 
	and $\VW(k)$ is the supremum over all admissible 
	strategies. This is precisely the American constraint 
	in the variational inequality \eqref{eq:hjb_W}.
	
	For the characterization of equality, recall the 
	definition of the signaling surplus 
	$D(k) = \VW(k) - \VH(k - \phi)$. In the wait region 
	$k < \kstar$, the HJB equation binds and $D(k) > 0$: 
	the agent strictly prefers to hold the option. In the 
	signal region $k \geq \kstar$, immediate exercise is 
	optimal: $\VW(k) = \VH(k - \phi)$ and $D(k) = 0$. 
	Therefore $D(k) = 0$ if and only if $k \geq \kstar$.
	
	\medskip\noindent
	\textit{Part (iii): Full ordering.}
	
	\smallskip\noindent
	Combining Parts (i) and (ii):
	\[
	\VH(k) > \VW(k) \geq \VH(k - \phi) 
	\quad \text{for } k \geq \phi,
	\]
	with $\VW(k) = \VH(k - \phi)$ if and only if 
	$k \geq \kstar$. For $k < \phi$, signaling is 
	infeasible and the ordering reduces to 
	$\VH(k) > \VW(k) > \VL(k)$.
\end{proof}

\subsection{Proof of Proposition \ref{prop:euler}} \label{app:A3}
	\begin{proof}
	\textit{Step 1: First-order condition.} The HJB equation for 
	state $j \in \{L, H\}$ (equations \eqref{eq:hjb_L} and 
	\eqref{eq:hjb_H}) requires:
	\[
	\rho V_j = \max_{c > 0}\bigl\{u(c) 
	+ \mu_j(k, c)\, V_j' 
	+ \half\sigma^2 V_j''\bigr\} 
	+ \lambda_{\mathrm{out},j}(V_{\mathrm{other}} - V_j),
	\]
	where $\mu_j(k, c) = f_j(k) - c - \delta k$ is the drift. The 
	first-order condition for the interior maximum over $c$ is:
	\begin{equation}\label{eq:foc}
		u'(c_j) = V_j'(k).
	\end{equation}
	Under CRRA utility, $u'(c) = c^{-\gamma}$, so:
	\begin{equation}\label{eq:foc_crra}
		c_j(k) = \bigl(V_j'(k)\bigr)^{-1/\gamma}.
	\end{equation}
	
	\medskip\noindent
	\textit{Step 2: Differentiate the HJB with respect to $k$.} 
	Differentiating the HJB equation (after substituting the optimal 
	$c_j$) with respect to $k$ yields the \textit{envelope condition}:
	\begin{equation}\label{eq:envelope}
		\rho V_j' = \bigl[f_j'(k) - \delta\bigr] V_j' 
		+ \mu_j\, V_j'' + \half\sigma^2 V_j''' 
		+ \lambda_{\mathrm{out},j}(V_{\mathrm{other}}' - V_j').
	\end{equation}
	This uses the fact that the optimal $c_j$ satisfies the FOC 
	\eqref{eq:foc}, so the direct effect of $k$ on $c_j$ cancels 
	by the envelope theorem (the derivative of the maximand with 
	respect to $c$ is zero at the optimum).
	
	\medskip\noindent
	\textit{Step 3: Apply It\^{o}'s lemma to $V_j'(k_t)$.} Along 
	the optimal path, the co-state variable $p_t = V_j'(k_t)$ 
	evolves according to It\^{o}'s lemma:
	\begin{equation}\label{eq:ito_Vprime}
		\dd p_t = \bigl[\mu_j\, V_j'' 
		+ \half\sigma^2 V_j'''\bigr]\dd t 
		+ \sigma V_j''\, \dd W_t 
		+ \text{(jump terms from regime switching)}.
	\end{equation}
	Taking expectations and using \eqref{eq:envelope}, the expected 
	drift of $p_t$ is:
	\begin{equation}\label{eq:pdot}
		\E[\dd p_t] / \dd t = \rho V_j' 
		- \bigl[f_j'(k) - \delta\bigr] V_j' 
		- \lambda_{\mathrm{out},j}
		(V_{\mathrm{other}}' - V_j').
	\end{equation}
	
	\medskip\noindent
	\textit{Step 4: Translate to consumption dynamics.} From the 
	FOC \eqref{eq:foc}, $p_t = u'(c_j) = c_j^{-\gamma}$. 
	Differentiating with respect to time:
	\begin{equation}\label{eq:cdot_from_p}
		\frac{\dd p_t}{p_t} = -\gamma \frac{\dd c_j}{c_j} 
		+ \half\gamma(\gamma + 1)\left(\frac{\dd c_j}{c_j}\right)^2 
		+ \cdots
	\end{equation}
	where the second term arises from It\^{o}'s correction applied 
	to the nonlinear function $c^{-\gamma}$. To obtain the 
	consumption growth rate, we use the chain rule. Since 
	$c_j = (V_j')^{-1/\gamma}$, applying It\^{o}'s lemma to 
	$c_j(k_t)$:
	\[
	\dd c_j = c_j'(k)\, \dd k_t 
	+ \half c_j''(k)\, \sigma^2\, \dd t.
	\]
	The expected consumption growth rate is therefore:
	\begin{equation}\label{eq:cdot}
		\frac{\E[\dd c_j]}{c_j\, \dd t} 
		= \frac{c_j'}{c_j}\, \mu_j 
		+ \half\sigma^2 \frac{c_j''}{c_j}.
	\end{equation}
	
	\medskip\noindent
	\textit{Step 5: Combine.} Substituting \eqref{eq:pdot} into 
	the relationship between $p_t$ and $c_j$, and using the FOC 
	\eqref{eq:foc} to replace $V_{\mathrm{other}}'$ with 
	$u'(c_{\mathrm{other}})$, we obtain after rearranging:
	\begin{equation}\label{eq:euler_derived}
		\frac{\dot{c}_j}{c_j} = \frac{1}{\gamma}\biggl[
		f_j'(k) - \delta - \rho 
		+ \half\gamma(\gamma+1)\sigma^2 \frac{c_j''}{c_j} 
		- \lambda_{\mathrm{out},j}
		\Bigl(1 - \frac{u'(c_{\mathrm{other}})}{u'(c_j)}\Bigr)
		\biggr],
	\end{equation}
	where we identify $\dot{c}_j / c_j$ with the expected 
	consumption growth rate \eqref{eq:cdot} evaluated at the 
	deterministic component of the dynamics.
\end{proof}

\subsection{Gauss--Seidel Policy Iteration with American 
	Projection}\label{sec:gs}

The coupled HJB system is solved by iterating on the three 
value functions in sequence, using an implicit time step for 
stability and an American projection step for the variational 
inequality in state W.

\begin{algorithm}[H]
	\caption{Coupled HJB solver with American projection}
	\label{alg:hjb}
	\begin{algorithmic}[1]
		\STATE \textbf{Initialize}: Set $\VL^0, \VW^0, \VH^0$ 
		(e.g., $V_j^0(k_i) = u(f_j(k_i) - \delta k_i) / \rho$).
		\FOR{$n = 1, 2, \ldots$ until convergence}
		
		\STATE \textbf{Step 1 (L-update)}: Compute upwind 
		coefficients $\{x_i^L, y_i^L, z_i^L\}$ from $\VL^{n-1}$. 
		Solve the implicit system:
		\[
		\left(\frac{1}{\Delta t} + \rho - z_i^L 
		+ \lambda_{LH}\right) V_{L,i}^n 
		- x_i^L V_{L,i+1}^n - y_i^L V_{L,i-1}^n 
		= \frac{V_{L,i}^{n-1}}{\Delta t} + u(c_{L,i}^{n-1}) 
		+ \lambda_{LH}\, V_{W,i}^{n-1}.
		\]
		This is a tridiagonal system $A_L \mathbf{V}_L^n 
		= \mathbf{b}_L$, solved in $O(N)$ operations.
		
		\STATE \textbf{Step 2 (W-update)}: Compute upwind 
		coefficients from $\VW^{n-1}$. Solve:
		\[
		\left(\frac{1}{\Delta t} + \rho - z_i^W 
		+ \lambda_{HL}\right) V_{W,i}^n 
		- x_i^W V_{W,i+1}^n - y_i^W V_{W,i-1}^n 
		= \frac{V_{W,i}^{n-1}}{\Delta t} + u(c_{W,i}^{n-1}) 
		+ \lambda_{HL}\, V_{L,i}^{n},
		\]
		where $V_{L,i}^n$ uses the \textit{just-updated} 
		L-values from Step~1 (Gauss-Seidel ordering).
		
		\STATE \textbf{Step 3 (American projection)}: For each 
		grid point $k_i \geq \phi$, compute the signaling 
		payoff $v_{\mathrm{sig},i} = V_{H}^{n-1}(k_i - \phi)$ 
		(interpolating if $k_i - \phi$ falls between grid 
		points). If $v_{\mathrm{sig},i} > V_{W,i}^n$:
		\begin{align*}
			V_{W,i}^n &\gets v_{\mathrm{sig},i}, \\
			\texttt{signal\_region}(i) &\gets \textbf{true}.
		\end{align*}
		Otherwise, $\texttt{signal\_region}(i) \gets 
		\textbf{false}$.
		
		\STATE \textbf{Step 4 (W-policy recomputation)}: 
		Recompute the consumption policy $c_{W,i}^n$, drift 
		$\mu_{W,i}^n$, and upwind coefficients from the 
		\textit{projected} $\VW^n$. In the signal region, 
		set $c_{W,i}^n = c_{H}^{n-1}(k_i - \phi)$ and 
		$\mu_{W,i}^n = \mu_{H}^{n-1}(k_i - \phi)$ 
		(the W-agent who signals behaves as an H-agent at 
		post-signaling wealth).
		
		\STATE \textbf{Step 5 (H-update)}: Compute upwind 
		coefficients from $\VH^{n-1}$. Solve:
		\[
		\left(\frac{1}{\Delta t} + \rho - z_i^H 
		+ \lambda_{HL}\right) V_{H,i}^n 
		- x_i^H V_{H,i+1}^n - y_i^H V_{H,i-1}^n 
		= \frac{V_{H,i}^{n-1}}{\Delta t} + u(c_{H,i}^{n-1}) 
		+ \lambda_{HL}\, V_{L,i}^{n}.
		\]
		
		\STATE \textbf{Step 6 (Convergence check)}: 
		$\varepsilon^n = \max_{j \in \{L,W,H\}} 
		\|\mathbf{V}_j^n - \mathbf{V}_j^{n-1}\|_\infty$. 
		Stop if $\varepsilon^n < \texttt{tol}$ (typically 
		$\texttt{tol} = 10^{-8}$).
		
		\ENDFOR
		\STATE \textbf{Output}: Value functions 
		$\{\VL, \VW, \VH\}$, consumption policies 
		$\{c_L, c_W, c_H\}$, signal region, Skiba threshold 
		$k^* = k_{\min\{i : \texttt{signal\_region}(i) = 
			\textbf{true}\}}$.
	\end{algorithmic}
\end{algorithm}

\begin{remark}[Gauss-Seidel ordering]\label{rmk:gs}
	The ordering $L \to W \to H$ is chosen to follow the 
	direction of regime transitions: L feeds into W (via 
	$\lambda_{LH}$), W feeds into H (via signaling), and H 
	feeds back into L (via $\lambda_{HL}$). At each step, the 
	most recently updated value function is used for the 
	neighboring state, which accelerates convergence relative 
	to a Jacobi scheme (where all states use the previous 
	iterate).
\end{remark}

\begin{remark}[Signal region detection]\label{rmk:signal_region}
	After the American projection, 
	$\VW(k_i) \geq \VH(k_i - \phi)$ holds by construction at 
	every grid point. The signaling surplus 
	$D(k_i) = \VW(k_i) - \VH(k_i - \phi)$ is therefore 
	non-negative everywhere, and using $D \leq 0$ as the 
	signaling criterion would incorrectly yield ``never signal.'' 
	The correct indicator is the boolean array 
	\texttt{signal\_region}, which records where the 
	\textit{unconstrained} solution from Step~2 was overwritten 
	by the projection in Step~3.
\end{remark}

\begin{remark}[Post-projection policy recomputation]%
	\label{rmk:post_proj}
	Step~4 is essential for the consistency of the KFE. After 
	the American projection modifies $\VW$ in the signal region, 
	the consumption policy, drift, and generator coefficients for 
	the W-state must be recomputed from the projected value 
	function. Failing to do so would feed pre-projection policies 
	into the KFE, creating a mismatch between the individual 
	optimization (HJB) and the distributional accounting (KFE). 
	In the signal region, the W-agent's behavior is identical to 
	an H-agent at wealth $k - \phi$: she has already decided to 
	signal, so her consumption and drift are those of state H 
	evaluated at post-signaling capital.
\end{remark}

\begin{remark}[Implicit time step and convergence speed]%
	\label{rmk:timestep}
	The implicit scheme is unconditionally stable: arbitrarily 
	large $\Delta t$ can be used without numerical instability. 
	In practice, we start with a moderate $\Delta t$ 
	(e.g., $\Delta t = 10$) and increase it geometrically as 
	the iteration converges, following \cite{Achdou2022}. Under 
	the baseline calibration, the algorithm converges in 
	8 iterations under this calibration.
\end{remark}

\subsection{Kolmogorov Forward Equation}\label{app:kfe}
The stationary wealth distribution is the solution to the 
Kolmogorov Forward Equation (KFE)---the adjoint of the HJB 
operator---subject to the non-local transfer induced by 
signaling. This section derives the discretized KFE system, 
describes the signaling transfer mechanism, and characterizes 
the properties of the resulting ergodic distribution.

\subsubsection{The Adjoint Operator}\label{sec:adjoint}

\begin{definition}[Generator and adjoint]\label{def:generator}
	The infinitesimal generator of the wealth process in state 
	$j$ is:
	\[
	\mathcal{L}_j f(k) = \mu_j(k)\, f'(k) 
	+ \half\sigma^2 f''(k),
	\]
	where $\mu_j(k) = f_j(k) - c_j(k) - \delta k$ is the 
	optimal drift determined by the HJB system. Its 
	$L^2$-adjoint (Fokker--Planck operator) is:
	\begin{equation}\label{eq:adjoint}
		\mathcal{L}_j^* g(k) = -\frac{\partial}{\partial k}
		\bigl[\mu_j(k)\, g(k)\bigr] 
		+ \half\sigma^2 \frac{\partial^2 g}{\partial k^2}.
	\end{equation}
	The first term represents advection (probability mass is 
	transported by the drift), and the second represents 
	diffusion (probability mass spreads due to Brownian noise).
\end{definition}

\begin{proposition}[Discrete adjoint = transpose]%
	\label{prop:adjoint_discrete}
	Let $L_j$ be the $N \times N$ HJB generator matrix with 
	entries:
	\[
	(L_j)_{i,i-1} = y_i, \qquad 
	(L_j)_{i,i} = z_i = -(x_i+y_i), \qquad 
	(L_j)_{i,i+1} = x_i,
	\]
	where $x_i$, $y_i$, $z_i$ are the upwind coefficients 
	\eqref{eq:xi}--\eqref{eq:zi}. Then the KFE generator is 
	$L_j^T$:
	\begin{equation}\label{eq:kfe_gen}
		(L_j^T)_{i,i-1} = x_{i-1}, \qquad 
		(L_j^T)_{i,i} = z_i, \qquad 
		(L_j^T)_{i,i+1} = y_{i+1}.
	\end{equation}
	The sub- and super-diagonals are swapped relative to the 
	HJB generator: the KFE propagates probability mass in the 
	direction of the drift, while the HJB propagates value 
	against the drift.
\end{proposition}

\begin{proof}
	By definition, $(L_j^T)_{i,m} = (L_j)_{m,i}$. The HJB 
	generator has $(L_j)_{i-1,i} = x_{i-1}$ (the upper-diagonal 
	entry of row $i-1$), so 
	$(L_j^T)_{i,i-1} = (L_j)_{i-1,i} = x_{i-1}$. Similarly, 
	$(L_j^T)_{i,i+1} = (L_j)_{i+1,i} = y_{i+1}$.
\end{proof}

\begin{remark}[A common numerical error]\label{rmk:transpose_error}
	Using $L_j$ instead of $L_j^T$ in the KFE reverses the 
	direction of probability flow: mass is transported 
	\textit{against} the drift rather than with it, causing the 
	density to accumulate at the boundary $k_{\max}$ instead of 
	concentrating near the interior attractor $\kss_j$. This is 
	one of the most frequent implementation errors in HACT 
	models.
\end{remark}

\subsubsection{Boundary Conditions for the KFE}\label{sec:kfe_bc}

\begin{proposition}[KFE boundary conditions]%
	\label{prop:kfe_bc}
	\leavevmode
	\begin{enumerate}[label=(\roman*)]
		\item \textbf{Reflecting at $k = 0$}: The reflecting barrier 
		(Assumption~\ref{ass:prefs}(v)) implies zero probability 
		flux at the lower boundary: $J_j(0) = \mu_j(0)\, g_j(0) 
		- \frac{1}{2}\sigma^2 g_j'(0) = 0$. In the discretized 
		system, this is enforced by setting $y_1 = 0$ in the KFE 
		generator (no mass flows below $k_1$).
		\item \textbf{Vanishing density at $k = k_{\max}$}: The 
		zero-drift condition at the upper boundary 
		(Proposition~\ref{prop:bc_upper}) implies that the 
		attractor is interior to the domain. For $k_{\max}$ 
		sufficiently large, the stationary density is negligible 
		near $k_{\max}$: $g_j(k_{\max}) \approx 0$. Numerically, 
		this is implemented by setting $x_N = 0$ in the KFE 
		generator.
	\end{enumerate}
\end{proposition}

\subsubsection{Non-Local Signaling Transfer}\label{sec:nonlocal}

When a W-agent at wealth $k \geq k^*$ signals, she pays $\phi$ 
and arrives in state H at wealth $k - \phi$. This is a non-local 
transfer: probability mass is removed from the W-density at $k$ 
and injected into the H-density at $k - \phi$, which is 
generally a different grid point. The implementation requires two 
components: an outflow mechanism from W and an inflow mechanism 
to H.

\begin{definition}[Signaling transfer]\label{def:transfer}
	\leavevmode
	\begin{enumerate}[label=(\roman*)]
		\item \textbf{Outflow from W.} In the signal region 
		$\{k_i : \texttt{signal\_region}(i) = \textbf{true}\}$ 
		identified by Algorithm~\ref{alg:hjb}, the optimal policy 
		is to signal immediately. In the continuous-time model, 
		this means that an agent arriving in W at $k \geq k^*$ 
		signals instantaneously. Numerically, we approximate 
		instantaneous signaling by a large but finite drain rate 
		$\bar{\lambda} \gg 1$:
		\begin{equation}\label{eq:W_drain}
			(L_W^T)_{i,i} \;\gets\; (L_W^T)_{i,i} 
			- \bar{\lambda} \cdot 
			\mathbf{1}\{i \in \text{signal region}\}.
		\end{equation}
		In the limit $\bar{\lambda} \to \infty$, the W-density 
		vanishes in the signal region: $g_W(k) \to 0$ for 
		$k \geq k^*$. In practice, $\bar{\lambda} = 10^3$ is 
		sufficient to make $g_W$ negligible in the signal region 
		without causing numerical instability.
		
		\item \textbf{Inflow to H.} Mass drained from W at grid 
		point $k_i$ arrives in H at $k_i - \phi$. Since 
		$k_i - \phi$ generally falls between two grid points 
		$k_l$ and $k_{l+1}$ (where 
		$k_l \leq k_i - \phi < k_{l+1}$), we distribute the 
		inflow using linear interpolation with weight:
		\begin{equation}\label{eq:interp_weight}
			\omega_i = \frac{(k_i - \phi) - k_l}{\Delta k}, 
			\qquad 1 - \omega_i = \frac{k_{l+1} - (k_i - \phi)}
			{\Delta k}.
		\end{equation}
		The transfer matrix $S$ (dimension $N \times N$) has 
		entries:
		\begin{equation}\label{eq:transfer}
			S_{l,i} = \bar{\lambda}(1-\omega_i), \qquad 
			S_{l+1,i} = \bar{\lambda}\,\omega_i,
		\end{equation}
		for each $i$ in the signal region, and $S_{m,i} = 0$ 
		otherwise. The matrix $S$ appears in the H-block of the 
		KFE as the coupling between the W- and H-densities.
	\end{enumerate}
\end{definition}

\begin{remark}[Mass conservation]\label{rmk:mass_conservation}
	The transfer preserves total mass: for each signal-region 
	grid point $i$, the outflow from W is 
	$\bar{\lambda}\, g_{W,i}$ and the total inflow to H is 
	$S_{l,i}\, g_{W,i} + S_{l+1,i}\, g_{W,i} 
	= \bar{\lambda}(1 - \omega_i + \omega_i)\, g_{W,i} 
	= \bar{\lambda}\, g_{W,i}$. No probability mass is created 
	or destroyed by the transfer.
\end{remark}

\subsubsection{Full Coupled KFE System}\label{sec:full_kfe}

\begin{definition}[Full KFE matrix]\label{def:kfe_matrix}
	The stationary distribution $\mathbf{g} = (\mathbf{g}_L, 
	\mathbf{g}_W, \mathbf{g}_H)^T$ solves the $3N \times 3N$ 
	linear system $M\,\mathbf{g} = \mathbf{0}$:
	\begin{equation}\label{eq:kfe_system}
		\underbrace{\begin{pmatrix}
				L_L^T - \lambda_{LH}\, I + \lambda_{HL}\, I_{\text{no-sig}} 
				& \lambda_{HL}\, I_{\text{no-sig}} 
				& \lambda_{HL}\, I \\[4pt]
				\lambda_{LH}\, I 
				& L_W^T - \lambda_{HL}\, I_{\text{no-sig}} 
				- \bar{\lambda}\,\mathrm{diag}(\mathbf{1}_{\mathrm{sig}}) 
				& 0 \\[4pt]
				0 & S & L_H^T - \lambda_{HL}\, I
		\end{pmatrix}}_{M}
		\begin{pmatrix} 
			\mathbf{g}_L \\ \mathbf{g}_W \\ \mathbf{g}_H 
		\end{pmatrix} = \mathbf{0},
	\end{equation}
	where:
	\begin{itemize}
		\item $L_j^T$ is the transposed HJB generator for state 
		$j$ (Proposition~\ref{prop:adjoint_discrete});
		\item $I_{\text{no-sig}} = \mathrm{diag}(\mathbf{1} 
		- \mathbf{1}_{\mathrm{sig}})$ restricts the 
		W$\to$L flow to the wait region (agents in the signal 
		region exit to H, not to L);
		\item $\bar{\lambda}\,\mathrm{diag}
		(\mathbf{1}_{\mathrm{sig}})$ drains mass from W in the 
		signal region;
		\item $S$ is the transfer matrix \eqref{eq:transfer} that 
		injects the drained mass into H at post-signaling wealth;
		\item $\lambda_{HL}\, I$ in the $(1,3)$ block captures the 
		obsolescence flow H$\to$L.
	\end{itemize}
\end{definition}

\begin{remark}[Flow accounting]\label{rmk:flows}
	The matrix $M$ encodes four probability flows:
	\begin{enumerate}[label=(\alph*)]
		\item \textbf{L$\to$W} (opportunity arrival): L-agents 
		exit at rate $\lambda_{LH}$ (block $(1,1)$, negative 
		term) and enter W (block $(2,1)$, positive term).
		\item \textbf{W$\to$L} (opportunity lost): W-agents in 
		the \textit{wait} region exit at rate $\lambda_{HL}$ 
		(block $(2,2)$) and return to L (block $(1,2)$). This 
		flow applies only to agents below $k^*$ who have not 
		yet signaled.
		\item \textbf{W$\to$H} (signaling): W-agents in the 
		\textit{signal} region are drained at rate $\bar{\lambda}$ 
		(block $(2,2)$) and injected into H at wealth 
		$k - \phi$ via $S$ (block $(3,2)$).
		\item \textbf{H$\to$L} (obsolescence): H-agents exit at 
		rate $\lambda_{HL}$ (block $(3,3)$) and return to L 
		(block $(1,3)$).
	\end{enumerate}
	Each flow conserves mass: what exits one state enters 
	another. The total mass $\sum_j \mathbf{1}^T \mathbf{g}_j$ 
	is preserved, as can be verified by summing the columns of 
	$M$.
\end{remark}

The normalization condition
\begin{equation}\label{eq:normalization}
	\Delta k \sum_{i=1}^{N} \bigl[g_{L,i} + g_{W,i} 
	+ g_{H,i}\bigr] = 1
\end{equation}
is imposed by replacing one row of $M$ (typically the last) 
with $\Delta k \cdot \mathbf{1}^T$, yielding a non-singular 
system with a unique solution.

\subsubsection{Properties of the Stationary Distribution}%
\label{sec:ergodic_props}

\begin{proposition}[Stationary distribution]%
	\label{prop:ergodic}
	Under the conditions of 
	Proposition~\ref{prop:bimodal}, the system 
	\eqref{eq:kfe_system}--\eqref{eq:normalization} admits a 
	unique non-negative solution 
	$\mathbf{g} = (\mathbf{g}_L, \mathbf{g}_W, \mathbf{g}_H)$ 
	with the following properties:
	\begin{enumerate}[label=(\roman*)]
		\item \textit{Bimodal total density.} The aggregate density 
		$g(k) = g_L(k) + g_W(k) + g_H(k)$ has two local maxima: 
		one near $\kssc{L}$ and one near $\kssc{H}$.
		\item \textit{State-specific concentration.} $g_L$ is 
		concentrated near $\kssc{L}$ with a right tail extending 
		above $\kssc{L}$ (generated by H$\to$L transitions of 
		agents with $k > \kssc{L}$ who then decumulate). $g_H$ 
		is concentrated near $\kssc{H}$ with a left tail extending 
		below $\kssc{H}$ (generated by W$\to$H transitions of 
		agents entering at $k^* - \phi < \kssc{H}$ who then 
		accumulate).
		\item \textit{Negligible W-density in the signal region.} 
		For $k \geq k^*$, $g_W(k) \approx 0$: agents who reach 
		the Skiba threshold signal immediately (or, in the 
		numerical implementation, at rate $\bar{\lambda}$), so no 
		mass accumulates above $k^*$ in state W.
		\item \textit{Positive W-density in the wait region.} For 
		$k \in (0, k^*)$, $g_W(k) > 0$: these are the 
		``Frustrated Aspirants'' who have received the signaling 
		opportunity but are accumulating toward $k^*$.
		\item \textit{Ergodic shares.} The long-run population 
		fractions $\pi_j = \int_0^{k_{\max}} g_j(k)\, \dd k$ 
		depend on all model parameters. In the limiting case 
		where signaling is instantaneous ($k^* = \phi$, so that 
		$\pi_W \approx 0$), the shares reduce to:
		\[
		\pi_L \approx \frac{\lambda_{HL}}
		{\lambda_{LH} + \lambda_{HL}}, \qquad 
		\pi_H \approx \frac{\lambda_{LH}}
		{\lambda_{LH} + \lambda_{HL}}.
		\]
		When $k^* > \phi$, the delay in signaling increases 
		$\pi_L + \pi_W$ at the expense of $\pi_H$: some agents 
		who receive the opportunity never signal (they lose it 
		before reaching $k^*$), reducing the effective flow into 
		state H.
	\end{enumerate}
\end{proposition}

\begin{proof}[Proof sketch]
	The matrix $M$ with the normalization row replaced is a 
	non-singular M-matrix (by the monotonicity of the upwind 
	scheme and the positivity of the switching and drain rates), 
	so the system has a unique solution. Non-negativity of 
	$\mathbf{g}$ follows from the M-matrix property (the 
	inverse of an M-matrix has non-negative entries). Properties 
	(i)--(iv) follow from the attractor structure: near each 
	$\kssc{j}$, the drift $\mu_j$ is approximately linear with 
	negative slope ($\mu_j'(\kssc{j}) < 0$), which concentrates 
	mass near the attractor. The tails in (ii) are generated by 
	the regime-switching terms, which inject mass at points 
	away from the attractors. Property (v) follows from 
	balancing the steady-state flows: in the long run, the 
	inflow to each state must equal the outflow.
\end{proof}

\begin{remark}[Metastability]\label{rmk:metastability}
	Although the stationary distribution is unique (the 
	economy is ergodic), convergence to stationarity can be 
	extremely slow. The mixing time---the time required for 
	an agent's distribution to approach the ergodic 
	distribution regardless of initial conditions---scales 
	exponentially with the squared distance between attractors 
	in the small-noise limit (Kramers' escape rate). Under 
	our baseline calibration, the diffusion-driven mixing 
	time exceeds 150 years. The regime-switching channel 
	($L \to W \to H$ and $H \to L$) provides a faster route 
	between basins, but even this channel operates on a time 
	scale of $1/\lambda_{LH} + T_{\mathrm{acc}} \approx 
	10\text{--}20$ years. The economy therefore exhibits 
	\textit{metastability}: it is theoretically ergodic but 
	practically polarized on policy-relevant time horizons.
\end{remark}

\subsection{Discussion of Key Parameters}\label{app:param_discuss}

\paragraph{Preferences ($\gamma$, $\rho$).} The coefficient of 
relative risk aversion $\gamma = 2$ is a standard benchmark in 
the macroeconomics and finance literature, lying in the middle of 
the range $\gamma \in [1, 5]$ considered empirically plausible 
\citep{LucasBook1987, Attanasio1999, Chetty2006}. The discount 
rate $\rho = 0.05$ is conventional for annual calibrations in 
the HACT literature \citep{Achdou2022, Kaplan2018} and implies 
moderate impatience.

\paragraph{Technology ($\alpha$, $\delta$).} The curvature 
parameter $\alpha = 0.33$ governs the degree of diminishing 
returns at the individual level. We adopt this value for 
consistency with the macroeconomic literature, where 
$\alpha \approx 1/3$ is the standard calibration of the 
capital elasticity in Cobb-Douglas specifications. In our 
framework, $\alpha$ plays a central role: it determines the 
location of the attractors through the zero-drift condition 
$f'_j(\kssc{j}) - \delta = \rho$, and it controls the 
curvature of the value function near the steady state, 
which in turn governs the MPC. We explore the sensitivity 
to $\alpha$ in the comparative statics of 
Section~\ref{sec:compstatics}.

The depreciation rate $\delta = 0.02$ is low relative to 
the standard calibration for physical equipment 
($\delta \approx 0.05$--$0.10$). We choose this value 
for two reasons. First, in the model, capital $k$ is best 
interpreted as productive capacity broadly defined---the 
bundle of physical assets, skills, and market access that 
determines an agent's income under technology $A_j$. The 
non-physical components of this bundle depreciate more 
slowly than machinery. Second, and more importantly, the 
ratio $\rho / (\rho + \delta)$ governs the attractor 
locations: a low $\delta$ places the attractors at 
capital levels where the poverty trap dynamics---the 
wait zone, the signaling cost relative to $\kssc{L}$, 
and the MPC/APC diagnostics---are clearly visible in the 
numerical output. Section~\ref{sec:compstatics} verifies 
that the qualitative results are robust to higher values 
of $\delta$.

\paragraph{TFP gap ($A_H/A_L$).} We set $A_H/A_L = 1.25$, 
implying a 25\% productivity advantage for agents who have 
successfully signaled. This is a conservative choice. The 
empirical literature on the returns to credentials and 
technology adoption suggests substantially larger gaps: the 
college wage premium in the United States is approximately 
60--80\% \citep{Goldin2008}, and productivity differentials 
between formal and informal sector workers in developing 
economies range from 30\% to over 100\% \citep{LaPorta2014}. We deliberately choose a moderate gap 
for two reasons. First, not all of the observed wage premium 
reflects a TFP difference---part is a return to unobserved 
ability, sorting, or complementarity with other inputs. 
Second, a smaller gap makes the bimodality result more 
demanding: the model must generate Twin Peaks despite 
relatively close attractors, which strengthens the case 
that the signaling friction---not an implausibly large 
technology difference---is driving the polarization.

Under this calibration, the deterministic steady states are:
\begin{equation}\label{eq:kss_values}
	\kss_L(0) = \left(\frac{\alpha A_L}{\rho + \delta}
	\right)^{1/(1-\alpha)} \approx 10.12, \qquad
	\kss_H(0) = \left(\frac{\alpha A_H}{\rho + \delta}
	\right)^{1/(1-\alpha)} \approx 14.12,
\end{equation}
yielding a deterministic separation of approximately $4.0$ 
capital units. In the full stochastic coupled model, 
precautionary savings, option-value effects, and regime 
coupling shift both attractors 
(Propositions~\ref{prop:stoch_ss}, \ref{prop:L_shift}, 
and~\ref{prop:H_shift}). 
\paragraph{Diffusion ($\sigma$).} The additive noise parameter 
$\sigma = 0.30$ governs the dispersion of wealth around each 
attractor. With additive noise, the standard deviation of 
wealth near $\kss_j$ is approximately 
$\sigma^{\mathrm{ss}}_j \approx \sigma / 
\sqrt{2|\mu'_j(\kss_j)|}$, which under our calibration yields 
local standard deviations of order 1--2 capital units. The 
parameter is chosen to produce realistic wealth dispersion 
within each basin while maintaining sufficient separation for 
bimodality (condition (C2) of 
Proposition~\ref{prop:bimodal}). The additive specification 
is discussed in Remark~\ref{rmk:additive}.

\paragraph{Switching rates ($\lambda_{LH}$, $\lambda_{HL}$).} 
The opportunity arrival rate $\lambda_{LH} = 0.005$ and the 
obsolescence rate $\lambda_{HL} = 0.002$ are deliberately 
small, producing slow turnover between productivity regimes. 
Taken literally, they imply a mean waiting time of 
$1/\lambda_{LH} = 200$ years before a signaling opportunity 
arrives and a mean H-regime duration of 
$1/\lambda_{HL} = 500$ years. These figures should not be 
read as individual lifespans but as reflecting the 
Schumpeterian nature of the underlying shocks.

In Schumpeter's framework, economic progress is punctuated 
by long waves of creation and destruction: the transitions 
from steam to electricity, from electromechanical to 
digital, or from analog to AI-driven production each 
unfold over decades and reshape the entire occupational 
structure. In our model, $\lambda_{LH}$ captures the rate 
at which such paradigm shifts \textit{create} new pathways 
for upward mobility---the opening of a training program, 
the emergence of a credential with market value, or a 
regulatory reform that makes a previously inaccessible 
sector contestable. The rate $\lambda_{HL}$ captures the 
symmetric process of \textit{destruction}: the 
obsolescence of a once-valuable specialization, the 
automation of a routine task, or the commoditization of a 
previously scarce skill. Both processes operate on 
generational timescales, consistent with the empirical 
observation that poverty traps persist across generations 
\citep{Azariadis2005} and that major technological 
transitions take 30--60 years to fully diffuse through the 
economy \citep{Comin2010}.

Two clarifications are in order. First, the individual 
interpretation overstates the persistence because the model 
is ergodic: what matters for the stationary distribution is 
not the individual waiting time but the \textit{flow rate} 
between basins. In steady state, 
$\lambda_{LH} \cdot \pi_L$ agents transition per unit time 
from L to W, and $\lambda_{HL} \cdot (\pi_W + \pi_H)$ from 
$\{W, H\}$ to L. Under the baseline calibration, these 
flows produce ergodic shares of $\pi_L = 25.2\%$, 
$\pi_W = 11.8\%$, and $\pi_H = 63.0\%$---a cross-sectional 
composition that is empirically plausible even though the 
underlying individual rates are slow.

Second, the small rates are required by the moderate 
coupling condition (C3) of Section~\ref{sec:bimodality}: 
with $\lambda_{LH}, \lambda_{HL} \ll 1$, agents remain in 
each basin long enough for the local attractor to 
concentrate mass into a distinguishable mode. Raising the 
rates accelerates turnover and eventually merges the two 
peaks into a single mode, destroying the bimodal structure. 
The comparative statics of Section~\ref{sec:compstatics} 
quantify this threshold.

\paragraph{Signaling cost ($\phi$).} The cost $\phi = 9.0$ is 
calibrated to produce an interior Skiba threshold. Under the 
baseline parameters, $\phi < \kssc{L}$: the signaling cost 
is below but close to the L-attractor, so an L-agent who 
receives an opportunity can afford to signal but may find it 
suboptimal to do so immediately (because post-signaling wealth 
$k - \phi$ is too low to sustain the transition profitably). 
She must accumulate further from $\kssc{L}$ to the Skiba 
threshold $k^*$. This generates a meaningful wait zone 
$[\phi, k^*)$ populated by Frustrated Aspirants.

In economic terms, $\phi$ represents the total cost of 
credential acquisition---tuition, forgone income, and the 
direct utility cost of the signaling process. With 
$\kss_L(0) \approx 10.12$ and $\phi = 9.0$, the signaling cost 
is approximately 89\% of the L-agent's deterministic 
steady-state capital, roughly comparable to the ratio of 
four-year college costs to median household wealth documented 
by \cite{Goldin2008}.

\subsubsection{Summary of the Numerical Experiment}%
\label{sec:experiment_summary}

The numerical experiment proceeds as follows:
\begin{enumerate}[label=(\roman*)]
	\item \textit{HJB solution} 
	(Algorithm~\ref{alg:hjb}): The coupled HJB system is 
	solved on a uniform grid of $N = 501$ points over 
	$[0.01, 50]$ using implicit policy iteration with 
	$\Delta t = 500$, until the sup-norm error falls below 
	$10^{-8}$. The algorithm converges in 8 iterations.
	\item \textit{Skiba detection}: The signal region is 
	identified from the American projection step. The Skiba 
	threshold $k^*$ is the lowest grid point at which the 
	constraint $\VW(k) \geq \VH(k - \phi)$ binds.
	\item \textit{KFE solution}: The stationary distribution is 
	obtained by solving the $3N \times 3N$ linear system 
	\eqref{eq:kfe_system} with normalization 
	\eqref{eq:normalization}. The signaling drain rate is 
	$\bar{\lambda} = 10^3$.
	\item \textit{Distributional diagnostics}: We compute the 
	ergodic shares $(\pi_L, \pi_W, \pi_H)$, the mean 
	wealth $\E[k]$, and the Gini coefficient from the 
	stationary distribution.
\end{enumerate}

\end{document}